\documentclass[11pt,reqno,a4paper,final]{amsart}
\usepackage[british]{babel}
\usepackage{amssymb,amstext,amsthm,eucal,bbm,amscd,bbold}
\usepackage[margin=2.1cm]{geometry}
\usepackage{array}
\usepackage[svgnames,table]{xcolor}
\usepackage[T1]{fontenc}
\usepackage[utf8]{inputenc}
\usepackage[small]{eulervm}
\usepackage{tgpagella}
\usepackage{mathrsfs}
\usepackage{verbatim}
\usepackage{pdfsync}
\usepackage[final]{microtype}
\usepackage[export]{adjustbox}
\usepackage{bm}
\usepackage{booktabs} 
\usepackage[font=small, labelfont=bf, skip=5pt]{caption}
\usepackage{appendix}
\usepackage{chngcntr}
\usepackage{apptools}
\usepackage{romanbar}
\usepackage{enumitem}
\usepackage{slashed}
\usepackage{comment}
\usepackage{import}
\usepackage{xifthen}

 \usepackage{mathtools}
 \usepackage{nicefrac}
\usepackage{braket}
 \usepackage{float}

\usepackage{tabularx}
\usepackage{multirow}
\renewcommand{\arraystretch}{1.3}
\newcommand{\PreserveBackslash}[1]{\let\temp=\\#1\let\\=\temp}
\newcolumntype{C}[1]{>{\PreserveBackslash\centering}p{#1}}
\newcolumntype{R}[1]{>{\PreserveBackslash\raggedleft}p{#1}}
\newcolumntype{L}[1]{>{\PreserveBackslash\raggedright}p{#1}}

\usepackage{tikz,pgfplots}
\pgfplotsset{compat=1.18}
\usepgfplotslibrary{fillbetween}
\usetikzlibrary{calc,arrows,arrows.meta,cd,shapes.geometric,positioning,through,decorations.markings,decorations.text}
\tikzset{>=latex}

\newcommand{\superbracket}[1]{[\kern-.15em[{#1}]\kern-.15em]}
\newcommand{\defeq}{\vcentcolon=}
\newcommand{\eqdef}{=\vcentcolon}

\newcommand{\Koszulsign}[3][]{(-1)^{|#2||#3| #1}}
\newcommand{\Bracket}[3]{[#1\,#2]_{#3}}
\newcommand{\normalord}[1]{\pmb{(}#1\pmb{)}}

\newcommand{\Tgh}{T^{\text{gh}}}
\newcommand{\Mgh}{M^{\text{gh}}}
\newcommand{\Tmat}{T^{\text{mat}}}
\newcommand{\Mmat}{M^{\text{mat}}}
\newcommand{\Ttot}{T^{\text{tot}}}
\newcommand{\Mtot}{M^{\text{tot}}}
\DeclareMathOperator{\End}{End}

\DeclareMathOperator{\reg}{reg}
\DeclareMathOperator{\ad}{ad}

\DeclareMathOperator{\im}{im}
\DeclareMathOperator{\Vir}{Vir}
\DeclareMathOperator{\Id}{Id}
\DeclareSymbolFont{bbold}{U}{bbold}{m}{n}
\DeclareSymbolFontAlphabet{\mathbbold}{bbold}

\newcommand{\dotr}{\mbox{$\boldsymbol{\cdot}$}}
\newcommand{\g}{\ensuremath{\mathfrak{g}}}
\newcommand{\W}{\mathcal{W}}

\newcommand{\semiinfforms}{\ensuremath{\Lambda^{\dotr}_\infty}}
\newcommand{\h}{\mathfrak{h}}
\newcommand{\z}{\mathfrak{z}}
\newcommand{\M}{\mathfrak{M}}
\newcommand{\bTil}{\Tilde{b}}
\newcommand{\cTil}{\Tilde{c}}
\newcommand{\beTil}{\Tilde{\beta}}
\newcommand{\gamTil}{\Tilde{\gamma}}

\makeatletter
\renewcommand*\env@matrix[1][\arraystretch]{%
  \edef\arraystretch{#1}%
  \hskip -\arraycolsep
  \let\@ifnextchar\new@ifnextchar
  \array{*\c@MaxMatrixCols c}}
\makeatother

\makeatletter
\newcommand{\medoplus}{\mathbin{\mathpalette\make@med\oplus}}
\newcommand{\medotimes}{\mathbin{\mathpalette\make@med\otimes}}
\newcommand{\make@med}[2]{%
  \vcenter{\hbox{%
    \scalebox{1.5}{$\m@th#1#2$}%
  }}%
}
\makeatother

\newcommand{\ZZ}{\mathbb{Z}}
\newcommand{\CC}{\mathbb{C}}
\newcommand{\NN}{\mathbb{N}}

\newtheoremstyle{myremarkstyle} 
        {5pt}                    
        {5pt}                    
        {}                   
        {}                           
        {\bfseries}                   
        {:}                          
        {.5em}                       
        {}  
\newtheorem{proposition}{Proposition}[section]
\newtheorem{lemma}[proposition]{Lemma}
\newtheorem{corollary}[proposition]{Corollary}
\newtheorem{theorem}[proposition]{Theorem}

\theoremstyle{myremarkstyle}
\newtheorem{remark}[proposition]{Remark}

\theoremstyle{definition}
\newtheorem{definition}[proposition]{Definition}

\newcommand{\subdefhighlight}[1]{\textit{\color{MidnightBlue}{#1}}}

\usepackage[backend=biber, style=trad-abbrv, sorting=none, url=false, isbn=false,backref=true]{biblatex}
\definecolor{highlights}{RGB}{193,0,67}
\usepackage{url}
\usepackage{hyperref}
\hypersetup{
  pdftitle   = {The BRST quantisation of chiral BMS-like field theories},
  pdfkeywords = {BMS, BRST, CFT, semi-infinite cohomology},
  pdfauthor  = {José Figueroa-O'Farrill, Girish Vishwa},
  colorlinks=true,
  linkcolor=NavyBlue,
  filecolor=magenta,      
  urlcolor=Navy,
  citecolor=DarkSlateBlue
}

\makeatletter
\g@addto@macro{\UrlBreaks}{\UrlOrds}   
\makeatother

\numberwithin{equation}{section}
\AtAppendix{\counterwithin{lemma}{section}}

\begin{document}

\title{The BRST quantisation of chiral BMS-like field theories}

\author[Figueroa-O'Farrill]{José M Figueroa-O'Farrill}
\author[Vishwa]{Girish S Vishwa}
\address{Maxwell Institute and School of Mathematics, The University
  of Edinburgh, James Clerk Maxwell Building, Peter Guthrie Tait Road,
  Edinburgh EH9 3FD, Scotland, United Kingdom}
\email[JMF]{\href{mailto:j.m.figueroa@ed.ac.uk}{j.m.figueroa@ed.ac.uk}, ORCID: \href{https://orcid.org/0000-0002-9308-9360}{0000-0002-9308-9360}}
\email[GSV]{\href{mailto:}{G.S.Vishwa@sms.ed.ac.uk}, ORCID: \href{https://orcid.org/0000-0001-5867-7207}{0000-0001-5867-7207}}

\begin{abstract}
  The BMS$_3$ Lie algebra belongs to a one-parameter family of Lie
  algebras obtained by centrally extending abelian extensions of the
  Witt algebra by a tensor density representation.
  In this paper we
  call such Lie algebras $\hat{\g}_\lambda$, with BMS$_3$ corresponding
  to the universal central extension of $\lambda = -1$. We construct the BRST complex for
  $\hat{\g}_\lambda$ in two different ways: one in the language of
  semi-infinite cohomology and the other using the formalism of vertex
  operator algebras. We pay particular attention to the case of
  BMS$_3$ and discuss some natural field-theoretical realisations. We
  prove two theorems about the BRST cohomology of $\hat{\g}_\lambda$.
  The first is the construction of a quasi-isomorphic embedding of the
  chiral sector of any Virasoro string as a $\hat{\g}_\lambda$ string.
  The second is the isomorphism (as Batalin--Vilkovisky algebras) of any
  $\hat{\g}_\lambda$ BRST cohomology and the chiral ring of a
  topologically twisted $N{=}2$ superconformal field theory.
\end{abstract}

\maketitle
\tableofcontents

\section{Introduction}\label{sec: intro}

Two-dimensional conformal field theories (2d CFTs) have been of
immense physical and mathematical interest ever since they were
discovered to appear as worldsheet descriptions of string theory and
various condensed matter systems. The rigorous algebraic formulation
of 2d CFT led to the birth of vertex operator algebras
\cite{Borcherds:1983sq}, whose significance in mathematics first
became prominent due to the construction of the monster vertex algebra
by Frenkel, Lepowsky and Meurman \cite{Frenkel:1988xz} and its usage
in the proof of the monstrous moonshine conjecture by Richard
Borcherds \cite{Borcherds1992MonstrousMA}.

The seminal paper \cite{Belavin:1984vu} by Belavin, Polyakov and
Zamolodchikov not only revealed the crucial role played by the
Virasoro algebra in such theories, but also generated a huge amount of
interest in the study of \emph{extended conformal algebras}. These are
the symmetry algebras of field-theoretic extensions of 2d CFTs,
obtained by adding a set of fields of some conformal weights, which
contain the Virasoro algebra as a subalgebra. The most celebrated and
well-known examples of these are the affine Kac-Moody algebras and the
superconformal algebras, which lie at the heart of Wess--Zumino--Witten
(WZW) models \cite{Witten:1983ar} that describe strings propagating on
Lie groups and superstrings respectively. However, the resulting
symmetry algebra after the extension need not be a Lie algebra; the
most notable example of this is the $W_3$ algebra, introduced by
Zamolodchikov \cite{Zamolodchikov:1985wn}. The BRST cohomology of the
$W_3$ algebra was studied in detail by Bouwknegt, McCarthy and Pilch
\cite{1996LNPM...42.....B}.

In recent years, there has been an increased interest in studying
certain abelian extensions of the Witt/Virasoro algebras and their
corresponding representation and field theories. In particular, from
the field theory perspective, the extension of 2d CFTs by a spin-2
quasiprimary field $M(z)$ which has regular operator product expansion
(OPE) with itself has garnered attention, as the symmetry algebra of
one such field theory, known as the BMS\textsubscript{3} algebra, was
shown to appear as the symmetry algebra\footnote{However, chiral 2d
  CFT techniques are only applicable in the ambitwistor setting (see
  \cite{Casali:2016atr} and \cite[Appendix A]{MasonSkinner}), since in
  general, the BMS\textsubscript{3} symmetry of the closed bosonic
  tensionless string worldsheet does not appear chirally
  \cite{Isberg:1993av, Bagchi:2013bga, Bagchi:2020fpr}. We address
  this caveat in our analysis.} of the tensionless closed bosonic string worldsheet \cite{Isberg:1993av, Bagchi:2013bga, Bagchi:2020fpr}. In mathematics, (a special case of)
this was introduced as the $W(2,2)$ algebra by Zhang and Dong
\cite{zhang2007walgebra}. This led to various works on the
representation theory of the BMS\textsubscript{3} algebra
\cite{Bagchi:2009pe, radobolja2013subsingular, JIANG2015118,
  Adamovic_2016, Jiang:2017vtt}. More recently, the representation
theory of numerous other extensions of the Virasoro algebra have been
studied. Some examples are the (twisted) Heisenberg-Virasoro algebra
\cite{tan2021simple}, mirror Heisenberg-Virasoro algebra
\cite{tan2021simple, gao2021nonweight}, $N=1$ super-BMS algebra
\cite{liu2023smooth}, BMS Kac-Moody algebra and Ovsienko--Roger
algebra \cite{Dilxat:2022zso}. What is noteworthy about all these
algebras is that they are all special cases (or their minimally
supersymmetric extensions) of the Lie algebra $\W(a,b)$, which is
constructed via the semi-direct sum of the Witt algebra $\W$ and
its tensor density modules $I(a,b)$, for $(a,b)\in\CC^2$ (see \cite{GaoJiangPei, FarahmandParsa:2018ojt}. Explicitly,
$\W(a,b) = \bigoplus_{n\in\ZZ} \big(\CC L_n \oplus \CC M_n\big)$ with
Lie bracket
\begin{equation}
    [L_n, L_m] = (n-m) L_{m+n} \quad [L_n, M_m] = -(a+m+ b n) M_{m+n} \quad [M_n, M_m] = 0.
\end{equation}
For our purposes, we restrict ourselves to $a,b\in\ZZ$, and since
$\W(a+a',b)\cong\W(a,b)$ for all $a'\in \ZZ$, it suffices to consider
the one-parameter family of Lie algebras\footnote{This one-parameter family of algebras was shown to appear as the near horizon symmetry algebra of non-extremal black holes in \cite{Grumiller:2019fmp}. In their work, the parameter $s$ comes from a choice of boundary condition, where $s=-\lambda$.} $\g_\lambda \defeq
\W(0,\lambda)$, where $\lambda\in\ZZ$.

In this paper, we consider 2d CFTs whose symmetry algebra is
$\g_\lambda$. As usual, this statement is merely a reformulation of
the Lie algebra $\g_\lambda$ in terms of fields in one formal variable
admitting certain OPEs. Nonetheless, with such a field-theoretic
formulation, we may then consider its BRST quantisation. In this
paper, we explicitly build the BRST operator for all $\g_\lambda$
field theories as the semi-infinite differential of the Lie algebra
$\g_\lambda$.

The notion that the BRST cohomology of various 2d CFTs coincides with
the semi-infinite cohomology of the underlying symmetry Lie algebra,
relative to its centre, is not a new one, particularly since the work
of Frenkel, Garland and Zuckerman \cite{frenkel1986semi}, in which
they explicitly computed the spectrum of the bosonic string as the
(relative) semi-infinite cohomology of the Virasoro algebra with
values in the Fock module. This interplay between the physical
spectrum of states and an algebraic structure corresponding to Lie
algebras has been a useful and powerful tool in both mathematics and
physics. The purely algebraic approach to BRST cohomology through the
construction of the semi-infinite wedge representation of the Lie
algebra at hand is very instructive in letting us build a free field
realisation of that algebra in terms of fermionic $bc$-systems. This
is done by repackaging the findings from semi-infinite representation
theory into generating functions of one variable. While this
formulation may simply be regarded as a trick which allows one to use
OPEs instead of the cumbersome infinite sums in semi-infinite
cohomology theory to perform mathematical computations, it also admits
a natural field-theoretic interpretation, where those generating
functions are precisely the fields that generate the 2d CFT of the
$bc$-systems. On the other hand, these $bc$-systems would be the
Faddeev--Popov ghosts that one would introduce in the BRST quantisation
of a 2d CFT whose symmetry algebra is the (central extension of the)
Lie algebra with which we started.

This approach of recasting mode algebras as fields will underpin the
entire paper, which is structured as follows. Section \ref{sec: Setup}
will introduce definitions for 2d CFTs and show the field-theoretic
formulation of $\g_\lambda$. In section \ref{sec: SIC and FFR}, we
review the notion of semi-infinite cohomology of $\ZZ$-graded
infinite-dimensional Lie algebras in general, and then construct the
semi-infinite wedge representation of $\g_\lambda$ explicitly to
demonstrate how fermionic $bc$-systems are simply the field-theoretic
formulation of semi-infinite cohomology. We then construct the BRST
operator of $\g_\lambda$ field theories, which indeed coincides with
the semi-infinite differential of $\g_\lambda$, and requires that the
matter sector of such theories have (Virasoro) central charge
$26+2(6\lambda^2-6\lambda+1)$ to be square-zero. Using these
constructions, we present a set of embedding theorems in section
\ref{sec: embedding theorems}, which relate the relative semi-infinite
cohomology (a.k.a. BRST cohomology) of the Virasoro and $\g_\lambda$
algebras and the chiral ring of a twisted $N=2$ superconformal field
theory (SCFT). These results present isomorphisms of homotopy
Batalin--Vilkovisky (BV) algebras, which are ``stronger'' than the
isomorphisms of graded vector spaces which one would expect from the
semi-infinite analogue of Shapiro's lemma (proven by Voronov in
\cite{voronov2002semiinfinite}). In section \ref{sec: BMS3}, we study
the special case of $\g_{\lambda=-1}$ in detail, since it is
isomorphic to the (centreless) BMS\textsubscript{3} algebra. We go
beyond semi-infinite representations and argue why a square-zero BRST
operator for the BMS\textsubscript{3} algebra cannot exist for
$c_M \neq 0$. We present two physical realisations of chiral
BMS\textsubscript{3} field theories - the ambitwistor string
\cite{MasonSkinner, Casali:2016atr} and the gauged Nappi--Witten string
\cite{Nappi:1993ie}. Finally, in section \ref{sec: Conclusion}, we
summarise our results for generic $\lambda\in\ZZ$, address the
implications of these results for BMS\textsubscript{3} field theories
(i.e., the case $\lambda=-1$), with reference to the caveat of
BMS\textsubscript{3} symmetry appearing in a non-chiral manner in
tensionless strings, and present some ongoing and potential extensions
to our work.

\section{Preliminaries}\label{sec: Setup}

In this section, we set up some notation and terminology with regards
to 2d CFTs and introduce the class of Lie algebras $\W(a,b)$. We then
show how to set up extended CFTs which admit $\g_\lambda \cong
\W(0,\lambda)$ symmetry, which will be the focus of this paper.

\subsection{Algebraic formalism for meromorphic 2d CFTs}
\label{sec:algebr-form-mero}
We refer the reader to the standard references \cite{Frenkel:1988xz, Ginsparg:1988ui, Thielemans:1994er, Schellekens:1996tg, DiFrancesco:1997nk, Polchinski:1998rq, schottenloher2008mathematical} on this subject for more details.

\begin{definition}
\label{def: VOA}
A \subdefhighlight{meromorphic 2d CFT} or \subdefhighlight{vertex
  operator algebra (VOA)} is given by the following data:
 \begin{enumerate}[label=(D\arabic*)]
 \item A complex vector superspace
   $V=\medoplus_{n\in\mathbb{Z}} V^{\Bar{0}}_n \oplus V^{\Bar{1}}_n$
   which has a $\ZZ_2$-grading and a $\ZZ$-grading that are compatible
   with each other, spanned by elements known as
   \subdefhighlight{states}. The $\ZZ$-grading is known as
   \subdefhighlight{conformal weight}. \label{item D1 space of states}
     \item An injective linear mapping sending a state $A\in V$ to a
       field $A(z)=\End V [[z,z^{-1}]]$ known as the
       \subdefhighlight{state-field correspondence}. For all $A\in
       V_h$, $A(z)=\sum A_n z^{-n-h}$. 
    \item A linear map $\partial\colon V_{h}\rightarrow V_{h+1}$ such
      that $(\partial A)(z)\defeq \frac{d}{dz}A(z)$.
    \item A set of bilinear brackets $[-,-]_n : V \otimes V
      \rightarrow V$, labelled by $n\in\mathbb{Z}$, defined by the
      \subdefhighlight{operator product expansion (OPE)}
    \begin{equation}
        \label{eq: OPE_def}
        A(z)B(w) = \sum_{n\ll\infty} \frac{[A\,B]_n (w)}{(z-w)^n}
    \end{equation}
    where the summation index $n\ll \infty$ indicates that there are
    only a finite number of singular terms (those with $n>0$) in the
    sum, satisfying
    \begin{itemize}
        \item \textbf{Identity:} There is a distinguished state
       called the \subdefhighlight{vacuum} $\mathbbm{1}\in V^{\Bar{0}}_0$ such that
       $\lim_{z\to 0} A(z) \mathbbm{1} = A$ for all $A\in
       V$ and $\partial \mathbbm{1}=0$. Thus, for all $A\in V$,
          \begin{equation*}
            [\mathbbm{1}\, A]_{n} = 
            \begin{cases}
              A,\quad n=0,\\
              0,\quad \text{otherwise}.
            \end{cases}
          \end{equation*}
        \label{item: Identity axiom CFT}
        \item \textbf{Commutativity:} For all $A,B\in V$,
        \begin{equation*}
            [A\,B]_n-\Koszulsign[+n]{A}{B}[B\,A]_n=\sum_{l\geq1}\frac{(-1)^{l+1}}{l!}\partial^l[A\,B]_{n+l}.
        \end{equation*}
        \label{item: commutativity CFT}
        \item \textbf{Associativity:} For all $A,B,C\in V$,
        \begin{equation*}
        \begin{split}
            \Bracket{\Bracket{A}{B}{p}}{C}{q} &= \sum_{l\geq q} (-1)^{l-q} \binom{p-1}{l-q} \Bracket{A}{\Bracket{B}{C}{l}}{p+q-1} \\
             - &\Koszulsign{A}{B} \sum_{l\geq 1} (-1)^{p-l} \binom{p-1}{l-1} \Bracket{B}{\Bracket{A}{C}{l}}{p+q-l}.
        \end{split}
        \end{equation*}
        \label{item: associativity CFT}
    \end{itemize} 
    The \subdefhighlight{normal ordered product} of two fields $A(z)$
    and $B(z)$ is their $n=0$ bracket and is denoted $[AB]_0\eqdef
    \normalord{AB}$. For convenience, we denote nested normal-ordered
    products as
    $\normalord{ABC}\defeq\normalord{A\normalord{BC}}$.\label{item D4}
    \item \label{item: D5 Vir element} A \subdefhighlight{Virasoro
        element} $T=\sum_{n\in\ZZ} L_n z^{-n-2} \in V^{\Bar{0}}_2$
      such that
    \begin{itemize}
        \item $[T\,T]_{n>4}=0$, $[T\,T]_4=\frac{1}{2}c_L\mathbbm{1}$,
          $[T\,T]_3= 0$, $[T\, T]_2 = 2T$ and $[T\,T]_1=\partial T$, where
          $c\in\CC$ is known as the  \subdefhighlight{central charge}.
        \item For all $A\in V_h$, $[T\,A]_2=hA$ and $[T\,A]_1=\partial
          A$. If in addition $[TA]_{n\geq 3}=0$, then $A(z)$ is a
          \subdefhighlight{primary field} with conformal weight $h$.
          If $[T\,A]_3 = 0$, but for some $n>3$, $[T\,A]_n \neq 0$,
          then $A(z)$ is a \subdefhighlight{quasiprimary field}.
    \end{itemize}
  \end{enumerate}
  From this point, meromorphic 2d CFTs as given in this definition
  will just be referred to as CFTs. CFTs admit rich mathematical
  structure and thus a myriad of useful properties to probe them. Some
  essential ones are listed and derived in appendix \ref{app:prop-merom-2d}. In particular, \ref{item: Mode algebra
    OPE relation} together with \ref{item: D5 Vir element} imply that
  the modes $\{L_n\}_{n\in\ZZ}$ obey the Virasoro algebra, and that
  $V$ contains a graded representation of the Virasoro algebra with
  central charge $c_L$, where the grading element $L_0\in\End V$ is
  diagonalisable (due to \ref{item D1 space of states}). If the field
  theory is generated by just $T$, we are in the usual case of
  non-logarithmic, meromorphic 2d CFT, whose symmetry algebra is that
  of the modes of the field $T(z)$ (i.e., the Virasoro algebra).
\end{definition}

Of course, using the formalism of definition \ref{def: VOA}, we can
consider field theories which are not generated by $T(z)$ alone, but
instead by $T(z)$ and an additional set of fields
$\{W_i(z)\}_{i\in I}$ of weight $h_i$, for some finite set $I$, with
OPEs that obey the axioms in \ref{item D4}. Such field theories are
still CFTs, since the generator of the conformal symmetry, $T(z)$, is
still one of the generators. The symmetry algebra of any such field
theory is then the mode algebra of the fields $T(z)$ and
$\{W_i(z)\}_{i\in I}$. It will contain the Virasoro algebra as a
subalgebra by construction. Hence, such CFTs are called
\subdefhighlight{extended CFTs.}

We now define a well-known example of CFTs called $bc$-systems
\cite{Friedan:1985ey, Friedan:1985ge}. In string theory, these often
appear in the ghost sector of the theory.
 
\begin{definition}
  \label{def: bc-systems general}
  A weight $(1-\lambda,\lambda)$ \subdefhighlight{$bc$-system} is a
  2d CFT formed from two bosonic (resp. fermionic) primary fields
  $b(z)$ and $c(z)$ of weights $1-\lambda$ and $\lambda$ with OPEs
  \begin{equation}
    b(z) c(w) = \frac{\mathbbm{1}(w)}{z-w}+\reg \iff [cb]_1=\epsilon,
  \end{equation}
  where $\epsilon=\pm1$ if $b(z)$ and $c(z)$ are fermionic (resp. bosonic).
  The Virasoro element is
  \begin{equation}
    \label{eq: Tbc general}
    T^{bc} = -\epsilon(1-\lambda) \normalord{b\partial c} + \epsilon \lambda \normalord{\partial b c}
  \end{equation}
  with central charge  $-\epsilon 2(6\lambda^2 - 6\lambda + 1)$. The
  fields admit the following mode expansions:
  \begin{equation}
    b(z) \defeq \sum_{n\in\ZZ}b_n z^{-n-(1-\lambda)} \quad\quad c(z) \defeq \sum_{n\in\ZZ} c_n z^{-n-\lambda} \label{eq: bc general mode exp}\\
  \end{equation}
  where $b_n$ and $c_n$ are endomorphisms of the underlying vector
  space of states of the CFT. Given a $bc$-system, a
  \subdefhighlight{vacuum state} $\ket{q}\in V$ of charge $q$, where
  $q\in\mathbb{Z}+\nicefrac{1}{2}$ for the NS sector and
  $q\in\mathbb{Z}$ for the R sector, is given by the conditions
  \begin{equation}
    \label{eq: bc_vacuum}
    \begin{split}
      &b_n \ket{q} = 0 \quad \quad n > \epsilon q - (1-\lambda) \\
      &c_n \ket{q} = 0\quad \quad n \geq - \epsilon q + (1-\lambda)
    \end{split}
  \end{equation}
  The space of states is built from the modes $b_n$ and $c_n$ that
  act non-trivially on the chosen vacuum.
\end{definition}
  
\begin{remark}
  Conventionally, bosonic $bc$-systems are referred to as
  $\beta\gamma$-systems. We will also adopt this convention. The
  corresponding space of states with vacuum choice $\ket{q}$ will be
  called $V^{\beta\gamma}_q$. For the fermionic $bc$-systems, the
  space of states will either be referred to as $V^{bc}$ or
  $\semiinfforms$, the latter notation denoting the space of
  semi-infinite forms, introduced in the next section.
\end{remark}

We will revisit $bc$-systems when we construct semi-infinite
cohomology and discuss BRST quantisation of extended CFTs.

\begin{definition}
  \label{def: g_lambda}
    The one-parameter family of \subdefhighlight{Lie algebras}
    $\g_\lambda \defeq\W(0,\lambda)$ has underlying vector space
    generated by $\{L_n, M_n\}_{n\in\ZZ}$ and is defined by the Lie bracket
    \begin{equation}
        \label{eq: g_lambda defn}
        [L_n, L_m] = (n-m) L_{m+n} \quad [L_n, M_m] = -(m+\lambda n) M_{m+n} \quad [M_n, M_m] = 0.
    \end{equation}
    Unless mentioned otherwise, we take $\lambda\in\ZZ$. Since there
    is a surjective Lie algebra homomorphism $\g_\lambda \rightarrow
    \W$ to the Witt algebra, we may pull back the
    \subdefhighlight{Gelfand-Fuchs cocycle}
    \begin{equation}
    \label{eq: Gelfand-Fuks cocycle}
        \gamma_V (L_n, L_m) = \frac{1}{12} n(n^2-1)\delta^0_{m+n}
    \end{equation}
    to $\g_\lambda$. This allows us to centrally extend $\g_\lambda$
    to $\hat{\g}_\lambda$, for all $\lambda\in\CC$. Explicitly, the
    Lie bracket on $\hat{\g}_\lambda$ is given by
    \begin{equation}
        \label{eq: ghat_lambda defn}
        \begin{split}
            [L_n, L_m] &= (n-m) L_{m+n} + \frac{1}{12}n(n^2-1)\delta^0_{m+n} c_L \\
            [L_n, M_m] &= -(m+\lambda n) M_{m+n} \\
            [M_n, M_m] &= 0.
        \end{split}
    \end{equation}
\end{definition}

\begin{remark}\label{rem:other-lambda}
    For $\lambda=-1,0,1$, there exist other possible central charges
    coming from other central extensions of these Lie algebras (see
    \cite[Theorem~2.3]{GaoJiangPei}). 
\end{remark}

\begin{definition}
\label{def: g_lambda field theory}
Let $\rho\colon\hat{\g}_\lambda \to \End V$ be a $\ZZ$-graded
$\hat{\g}_\lambda$-module\footnote{Usually, one works with modules in
  the Category $\mathcal{O}$ (see Definition~\ref{def:catO})}, where
$\rho(L_0)\in\End V$ is the grading element. Define generating
functions
\begin{equation*}
T(z) = \sum_{n\in\ZZ}\rho(L_n) z^{-n-2}\quad\text{and}\quad M(z) =
\sum_{n\in\ZZ}\rho(M_n) z^{-n-(1-\lambda)}.
\end{equation*}
These have the following OPEs:
\begin{align}
     T(z)T(w) &= \frac{1}{2}\frac{c_L\mathbbm{1}(w)}{(z-w)^4} + \frac{2T(w)}{(z-w)^2} + \frac{\partial T(w)}{z-w} + \reg. \label{eq: g_lambda TT OPE}\\
    T(z)M(w) &= \frac{(1-\lambda)M(w)}{(z-w)^2} + \frac{\partial M(w)}{z-w} + \reg. \label{eq: g_lambda TM OPE} \\
    M(z)M(w) &= \reg. \label{eq: g_lambda MM OPE}
\end{align}
Using \ref{item: Mode algebra OPE relation}, one can show that the
above OPEs are equivalent to the commutator of the modes $\rho(L_n)$
and $\rho(M_n)$ obeying the $\hat{\g}_\lambda$ algebra \eqref{eq: ghat_lambda defn}, with the central
element $c_L$ acting as some multiple of the identity on the space
$V$. For convenience, this multiple is also called $c_L$, which is
then referred to as the central charge of the representation
\ref{item: D5 Vir element}. Thus, a
\subdefhighlight{$\hat{\g}_\lambda$ field theory} is any CFT generated
by fields $T(z)$ and $M(z)$ which admit the OPEs \eqref{eq: g_lambda
  TT OPE}, \eqref{eq: g_lambda TM OPE} and \eqref{eq: g_lambda MM
  OPE}. Equivalently, $\hat{\g}_\lambda$ field theories are extended
CFTs with symmetry algebra $\hat{\g}_\lambda$.
\end{definition}

\section{Semi-infinite cohomology and free field realisations}\label{sec: SIC and FFR}

This section aims to elucidate the relationships between the BRST
quantisation of extended CFTs whose symmetry algebras are Lie algebras
and the relative semi-infinite cohomology of the underlying symmetry
algebra. For brevity, we will drop the ``relative'' once this notion
is introduced. We will write down free field realisations of
$\g_\lambda$ in terms of both fermionic and bosonic $bc$-systems. As
we will show, the former is simply the field-theoretic formulation of
the semi-infinite wedge representation of $\hat{\g}_\lambda$; the
reason for needing the ``hat'' will be clarified in this section. It
will also enter our expression for the BRST current, whose zero mode
(i.e BRST operator of the $\hat{\g}_\lambda$ field theory) coincides
with the semi-infinite differential of $\hat{\g}_\lambda$.

\subsection{Review of semi-infinite cohomology}

We provide a review of semi-infinite cohomology as explained in
\cite{frenkel1986semi} and
\cite{He_2018_RemarkonSemiInfCohomology}. Some key proofs are provided
in appendix~\ref{app:some-proofs-semi}. After building up the
framework in general, we show the explicit computations of the
semi-infinite wedge representation of $\g_\lambda$.

\subsubsection{Building the space of semi-infinite forms}

Let $\g = \medoplus_{n\in\mathbb{Z}} \g_n$ be a $\ZZ$-graded Lie
algebra over $\mathbb{C}$, with $\dim \g_n < \infty\ \ \forall
n\in\mathbb{Z}$. Let $\g_{\pm} \defeq \medoplus_{\pm n>0} \g_n$. Let
$\{e_i\}_{i\in\mathbb{Z}}$ be a basis for $\g$ such that if $e_i \in
\g_n$ (for some $i,n\in\mathbb{Z}$), then either $e_{i+1}\in \g_n$ or
$e_{i+1}\in \g_{n+1}$. Let $\g' = \medoplus_{n\in\mathbb{Z}} \g'_n$
be the restricted dual of \g\ with $\g'_n = \g^*_n = \text{Hom}(\g_n,
\mathbb{C})$. Let $\{e'_i\}_{i\in\mathbb{Z}}$ be the dual basis for
$\g'$, where $e'_i(e_j) = \delta_{ij}$.

We may define a Clifford algebra $\text{Cl}(\g \medoplus \g')$ with
respect to the dual pairing $\langle-,-\rangle : \g' \times \g \to
\mathbb{C}$, defined by $\left<x',x\right> \defeq x'(x)$, as follows. For any
$x+x' \in \text{Cl}(\g \medoplus \g')$, we have the following relation
between the product "$\cdot$" of the algebra and the dual pairing:
\begin{equation}
    (x+x')\cdot(x+x') \eqdef (x+x')^2 = \langle x', x\rangle 1.
\end{equation}
Polarising this relation we obtain, for a more general combination of
elements,
\begin{equation}
    (a+b')\cdot(c+d') + (c+d')\cdot(a+b') = \langle d', a \rangle 1 + \langle b', c \rangle 1.
\end{equation}

\begin{definition}
\label{def: semiinfforms}
The \subdefhighlight{space of semi-infinite forms}
$\Lambda^{\dotr}_\infty$ is the space spanned by monomials
\begin{equation}
  \omega \defeq e'_{i_1} \wedge e'_{i_2} \wedge \dots
\end{equation}
where $i_1 > i_2 > \dots$ and $\exists$ $N(\omega) \in \mathbb{Z}$
such that $i_{k+1} = i_k - 1$ $\forall k > N(\omega)$.
\end{definition}

\begin{definition}
  For all $x\in\g$, $x'\in\g'$, we define the
  \subdefhighlight{contraction} $\iota(x)\in\End \semiinfforms$ and
  \subdefhighlight{exterior or wedge product}
  $\varepsilon(x')\in\End\semiinfforms$ through their actions on
  monomials as follows:
  \begin{align}
    \iota(x) e'_{i_1} \wedge e'_{i_2} \wedge \dots &= \sum_{k\geq1} (-1)^{k-1} \langle x,e'_{i_k}\rangle e'_{i_1} \wedge e'_{i_2} \wedge \dots \wedge \widehat{e'_{i_k}} \wedge\dots \label{eq: defn_contraction},\\
    \varepsilon(x') e'_{i_1} \wedge e'_{i_2} \wedge \dots &= x' \wedge e'_{i_1} \wedge e'_{i_2} \wedge \dots \label{eq: defn_extprod}
  \end{align}
  where the hat denotes omission.
\end{definition}

The following is the result of a simple calculation.

\begin{lemma}
\label{lem: fundamental_anticomms}
For all $x,y\in\g$ and $x',y'\in\g'$, the following (anti)commutation relations hold:
\begin{equation}
\label{eq: contr_ext_anticomms}
\begin{split}
     [\iota(x), \iota(y)] &= [\varepsilon(x'),\varepsilon(y')] = 0\\ 
     [\iota(x), \varepsilon(x')] &= \langle x', x \rangle \Id_{\semiinfforms}.
 \end{split}  
\end{equation}
\end{lemma}

\begin{proposition}\label{prop:cliff-module}
  $\Lambda^{\dotr}_\infty$ admits a Clifford module structure over $\textup{Cl}(\g \medoplus \g')$.
\end{proposition}

\subsubsection{Constructing the semi-infinite wedge representation}

\begin{definition}
Just like with any other Lie algebra, we can define the \subdefhighlight{adjoint representation} of \g \ via the linear map
$\ad\colon \g \to\End(\g)$ 
\begin{equation}
\label{eq: adjoint REP}
    \ad_x \defeq [x,-] \quad \forall x\in\g.
\end{equation}
Similarly, the \subdefhighlight{coadjoint representation} of $\g$ is given by the linear map $\ad'\colon \g \to \End(\g')$ such that $\forall x,y\in\g$, $y'\in\g'$,
\begin{equation}
\label{eq: coadjoint REP}
    (\ad'_x y')(y) \defeq -\langle y', \ad_x \rangle,
\end{equation}
so $\ad'_x y' = \langle y', [x,-]\rangle \in \g'$ indeed.
\end{definition}
The most natural guess for a representation $\rho\colon\g\to\End\semiinfforms$ is the generalisation of the coadjoint action to semi-infinite monomials:
\begin{equation}
\label{eq: rho_naive}
    \rho(x) e'_{i_1} \wedge e'_{i_2} \wedge \dots = \sum_{k\geq1} e'_{i_1} \wedge e'_{i_2} \wedge \dots \wedge \ad'_x e'_{i_k} \wedge \dots = \sum_{k\geq 1} \varepsilon(\ad'_x e'_{i_k}) \iota(e_{i_k}) e'_{i_1} \wedge e'_{i_2} \wedge \dots
\end{equation}
\begin{proposition}
\label{prop: rho_intext_commrels}
The following commutation relations hold for all $x,y\in \g$, $y'\in\g'$:
\begin{equation}
    \label{eq: rho_intext_commrels}
    [\rho(x), \iota(y)] = \iota(\ad_x y)  \quad \quad [\rho(x), \varepsilon(y')] = \varepsilon(\ad'_x y').
\end{equation}
\end{proposition}
Equation \eqref{eq: rho_naive} is well-defined except for $x \in
\g_0$, in which case, the infinite sum does not truncate to a finite
one. For a sensible definition of $\rho\colon \g \to
\End(\semiinfforms)$, we would like to fix the bad behaviour for
$x\in\g_0$. We start by defining a vacuum semi-infinite form
$\omega_0$ such that for all $x\in \g_n$, $y\in\g_{-n}$ where
$n\in\mathbb{Z}\setminus\{0\}$, $\rho([x,y]) \omega_0$ is proportional
to $\omega_0$.
The standard way to construct such a vacuum is by choosing $i_0$ such
that $e'_{i_0} \in \g_{m_0} \implies e'_{i_0+1} \in \g_{m_0+1}$ and
then letting
\begin{equation}
    \label{eq: def_omega0}
    \omega_0 \defeq e'_{i_0} \wedge e'_{i_0 -1} \wedge e'_{i_0 -2} \wedge \dots
\end{equation}
Hence, $\omega_0$ is the ordered wedge product of the dual basis elements spanning $\medoplus_{n\leq m_0} \g'_n$, for some $m_0\in\ZZ$.
Then for a given $\omega_0$, choose a $\beta \in \g'_0$ such that
$\beta([\g_0,\g_0]) = 0$, and define $\rho(x) \omega_0 = \langle \beta,x \rangle \omega_0$ for all $x\in\g_0$. By demanding that the
anti-commutation relations \eqref{eq: rho_intext_commrels} hold, we
may extend such an action of $\rho$ to all of
$\g$. Explicitly\footnote{This is not $\rho$ as defined in
  \cite{frenkel1986semi}, but we will prefer this version (also used
  in \cite{He_2018_RemarkonSemiInfCohomology} and \cite{AKMAN1993194})
  as it simplifies many explicit calculations since it is easier to
  work with the adjoint rather than the coadjoint
  action.}:
\begin{equation}
\label{eq: rho_semiinfforms}
    \rho(x) = \sum_{i\in\mathbb{Z}} :\iota(\ad_x e_i) \varepsilon(e'_i): +\ \langle \beta, x \rangle, 
\end{equation}
where we have defined the \subdefhighlight{normal-ordered product} with respect to the vacuum $\omega_0$ as
\begin{equation}
\label{eq: normal-ordered product SIC}
     :\iota(\ad_x e_i) \varepsilon(e'_i): =
     \begin{cases}
        \iota(\ad_x e_i) \varepsilon(e'_i), &\quad i\leq i_0\\
        -\varepsilon(e'_i)\iota(\ad_x e_i) , &\quad i>i_0
     \end{cases}.
\end{equation} 
Note that for all $x\in\g_n$ and $y\in\g_{-n}$ for $n\neq 0$, we have
$$\rho([x,y]) \omega_0 = \langle \beta, [x,y] \rangle \omega_0 = -\partial \beta (x, y) \omega_0,$$ where $\partial$ is the differential in Lie algebra cohomology. Hence, the infinite sums have indeed been tamed and, more specifically, $\rho(x,y)\omega_0$ is proportional to $\omega_0$ up to a factor determined by some coboundary. 
This is more than just an observation. It is closely related to the extent to which $\rho\colon\g\to\End(\semiinfforms)$ fails to be a Lie algebra representation, characterised by the following proposition.
\begin{proposition}[{\cite[Proposition 1.1]{frenkel1986semi}}]
\label{prop: FGZ_prop_1.1}
There exists a two-cocycle $\gamma\in H^2(\g)$ depending on the choice of vacuum $\omega_0$ and $\beta$ such that
\begin{enumerate}
    \item $\gamma(\g_m,\g_n) =  0$ $\forall m+n \neq 0$
    \item $[\rho(x), \rho(y)] = \rho([x,y]) +\gamma(x,y)$.
\end{enumerate}
If $\gamma$ is a coboundary, then there exists a choice of $\beta$ for a given $\omega_0$ such that $\gamma=0\in \Lambda^2(\g)$.
\end{proposition}

Proposition \ref{prop: FGZ_prop_1.1} tells us that the failure of \eqref{eq: rho_semiinfforms} to be a representation is characterised by a cocycle that is non-trivial in cohomology. This is in line with the fact that any failure that is characterised by a coboundary could be absorbed by an appropriate choice of $\beta$.
Also note that $\gamma$ is non-zero only on the zero-graded part of $\g \times \g$. This is also to be expected; for $x\notin\g_0$, \eqref{eq: rho_semiinfforms} reduces to the generalised coadjoint action, so one should not expect it to fail as a representation outside the zero-graded part.

If $\gamma$ is a representative of a non-trivial class in $H^2(\g)$, we are obstructed from making $\semiinfforms$ a $\g$-module. This obstruction can only be overcome by instead working with $\Hat{\g}$, the central extension of $\g$ constructed using $\gamma$. This allows us to view $\gamma$ as a coboundary instead, which can then be set identically to zero by an appropriate choice of $\beta$, as stated in proposition \ref{prop: FGZ_prop_1.1}. 

\subsubsection{Gradings}

There exist two natural gradings one can define on \semiinfforms.
\begin{definition}
\label{def: ghost number}
$\forall x\in \g$, $x'\in\g'$, 
\begin{equation}
    \text{Deg}\,\iota(x) =  - 1 \quad \quad \text{Deg}\,\varepsilon(x') = 1.
\end{equation}
Fixing $\text{Deg}\,\omega_0 \in \mathbb{Z}$, this defines the  grading \subdefhighlight{Deg} on \semiinfforms. We will sometimes refer to this grading as the \textit{ghost number}, the name being motivated by BRST quantisation in physics. 
\end{definition}
Since $\text{Deg}\,\rho = 0$, this makes $\Lambda^m_\infty \defeq \{\omega \in \semiinfforms\ |\ \text{Deg}\,\omega = m\}$ a \g-module $\forall m \in \mathbb{Z}$.
\begin{definition}
\label{def: deg}
$\forall x \in \g_n$, $x'\in\g'_n$,
\begin{equation}
    \deg \iota(x) = n \quad \quad \deg \varepsilon(x') = -n.
\end{equation}
Fixing $\deg\omega_0 \in \mathbb{Z}$, this defines the grading \subdefhighlight{deg} on \semiinfforms. In the context of CFT, this is referred to as the \textit{conformal weight}.
\end{definition}
Let $\Lambda^{m;n}_\infty \defeq \{\omega\in\Lambda^m_\infty\ |\ \deg\omega = n\}$ and $\Lambda^{\dotr;n}_\infty \defeq \{\omega\in\semiinfforms\ |\ \deg\omega = n\}$. For all $x\in\g_k$, $\rho(x)\colon \Lambda^{m;n}_\infty \to \Lambda^{m;n+k}_\infty$. Hence, $\deg$ makes $\Lambda^m_\infty$ and $\semiinfforms$ graded \g-modules.

\begin{definition}\label{def:catO0}
The \subdefhighlight{category $\mathcal{O}_0$} comprises graded
\g-modules $\mathfrak{M} = \medoplus_{n\in\mathbb{Z}} \mathfrak{M}_n$
such that $\dim\mathfrak{M}_n < \infty$ and for all $n>n_0$,
$\dim\mathfrak{M}_n = 0$, for some $n_0\in \mathbb{Z}$.
\end{definition}
\vspace{-7pt}
Regardless of how $\deg \omega_0$ is fixed, the structure of
\semiinfforms\ and the construction of $\deg$ is such that $\dim
\Lambda^{*;n}_\infty < \infty$ and is zero for all $n> n_0$ for some
$n_0\in\mathbb{Z}$. Hence, $\Lambda^{*;n}_\infty\in \mathcal{O}_0$.

\begin{definition}\label{def:catO}
The \subdefhighlight{category $\mathcal{O} \supset \mathcal{O}_0$}
comprises graded \g-modules $\mathfrak{M} = \medoplus_{n\in\mathbb{Z}}
\mathfrak{M}_n$ such that the $\g_+$-submodule $\{\mathcal{U}(\g_+)v\
|\ v\in\mathfrak{M}\}$, where $\mathcal{U}(\g_+)$ denotes the
universal enveloping algebra of $\g_+$, is finite dimensional for any
$v\in\mathfrak{M}$. We often abbreviate this last condition by saying
that the $\g_+$-action is locally nilpotent.
\end{definition}

\subsubsection{Semi-infinite complex}

Consider an arbitrary graded $\g$-module
$\mathfrak{M}\in\mathcal{O}_0$ with representation $\pi\colon \g \to
\End\mathfrak{M}$. Let $\deg v = n$ for all $v \in
\mathfrak{M}_n$. Defining $\deg(v \otimes \omega) \defeq \deg v + \deg
\omega$ turns $\mathfrak{M} \medotimes \semiinfforms$ into a
$\ZZ$-graded vector space, with each graded subspace being finite
dimensional. Then $\theta\colon \g \to \End(\mathfrak{M} \medotimes
\semiinfforms)$ given by $\theta(x) = \pi(x) + \rho(x)$ makes
$\mathfrak{M}\medotimes \semiinfforms$ a module in category $\mathcal{O}_0$.
\begin{definition}
The \subdefhighlight{semi-infinite differential} $d$ is given by 
\begin{equation}
    \label{eq: semi-infinite differential}
    d \defeq \sum_{i\in\mathbb{Z}} \pi(e_i)\varepsilon(e'_i) + \sum_{i<j} :\iota([e_i,e_j]) \varepsilon(e'_j) \varepsilon(e'_i):. 
\end{equation}
\end{definition}
\begin{proposition}
\label{prop: d^2=0}
 $d^2 = 0$.
\end{proposition}
This can be proven by using the result by Akman \cite{AKMAN1993194}
that the statement of proposition \ref{prop: d^2=0} is equivalent to
the representation $\theta\colon \g \to \End(\mathfrak{M} \medotimes
\semiinfforms)$ being given by
\begin{equation}
\label{eq: theta_as_anticomm}
    \theta(x) = [d,\iota(x)].
\end{equation}
Furthermore, the proof of the nilpotence of the semi-infinite
differential is one of the most illuminating examples for highlighting
the computational power of the OPE-oriented field-theoretic
formulation of semi-infinite cohomology when working with specific
examples of $\g$. In our case, we will be working with
$\g=\g_\lambda$, and the nilpotence of the semi-infinite differential,
which in the relative subcomplex to be defined below will be referred
to as the \subdefhighlight{BRST operator}, is shown by an OPE
computation that is much simpler than that done with the mode
expansion \eqref{eq: semi-infinite differential}.
\begin{definition}
\label{def: semiinfcohom}
$\{\mathfrak{M} \medotimes \semiinfforms, d\}$ is a (graded) complex 
\[ 
\begin{tikzcd}
    \dots \arrow[r, "d"] & \mathfrak{M} \medotimes \Lambda^{m-1}_\infty \arrow[r, "d"] & \mathfrak{M} \medotimes \Lambda^{m}_\infty \arrow[r, "d"] & \mathfrak{M} \medotimes \Lambda^{m+1}_\infty \arrow[r, "d"] &\dots 
\end{tikzcd} 
\]
and the corresponding cohomology $H^{\dotr}_\infty(\g;\mathfrak{M})$ is known as the \subdefhighlight{semi-infinite cohomology of \g} with values in $\mathfrak{M}$. Explicitly,
\begin{equation}
    H^m_\infty (\g; \mathfrak{M}) = \frac{\ker \big(d \colon \mathfrak{M} \medotimes \Lambda^{m}_\infty \to \mathfrak{M} \medotimes \Lambda^{m+1}_\infty\big)}{\im \big(d \colon \mathfrak{M} \medotimes \Lambda^{m-1}_\infty \to \mathfrak{M} \medotimes \Lambda^{m}_\infty\big)}.
\end{equation}
\end{definition}
The differential raises Deg by 1 and leaves $\deg$ unchanged, so one
can consider the complex for each deg too
\[
\begin{tikzcd}
    \dots \arrow[r, "d"] & \big(\M \medotimes \Lambda^{m-1}_\infty\big)^n \arrow[r, "d"] & \big(\M \medotimes \Lambda^{m}_\infty\big)^n \arrow[r, "d"] & \big(\M \medotimes \Lambda^{m+1}_\infty\big)^n \arrow[r, "d"] &\dots 
\end{tikzcd} 
\]
Then $H^m_\infty (\g; \mathfrak{M}) = \medoplus_{n\in\mathbb{Z}} H^{m;n}_\infty (\g; \mathfrak{M})$, where 
\begin{equation}
    H^{m;n}_\infty (\g; \mathfrak{M}) = 
    \frac{\ker \Big(d \colon \big(\M \medotimes \Lambda^{m}_\infty\big)^n \to \big(\M \medotimes \Lambda^{m+1}_\infty\big)^n \Big)}
    {\im \Big(d \colon \big(\M \medotimes \Lambda^{m-1}_\infty\big)^n \to \big(\M \medotimes \Lambda^{m}_\infty\big)^n \Big)}.
\end{equation}
As mentioned previously, what we refer to as ``semi-infinite cohomology'' is actually relative semi-infinite cohomology, which we define next.

\subsubsection{The relative subcomplex}

Let $\h\subset\g_0$ be a subalgebra.
We define a subspace
\begin{equation}
    C^{\dotr}_\infty(\g,\h;\mathfrak{M}) \defeq \{w\in\mathfrak{M} \medotimes \semiinfforms\ |\ \iota(x)w = 0\text{ and }\theta(x)w= 0\quad \forall x \in \h \}.
\end{equation}
Equation \eqref{eq: theta_as_anticomm} implies
\[ \theta(x) w = 0 \iff  \big(d\iota(x) + \iota(x)d\big) w = \iota(x)dw = 0 \ \ \forall w \in C^{\dotr}_\infty(\g,\h;\mathfrak{M}). \]
Consequently, for any $w\in C^{\dotr}_\infty(\g,\h;\mathfrak{M})$, $\iota(x) dw = 0$ and $\theta(x) dw = \big(d\iota(x) + \iota(x)d\big) dw = 0$, so $d\big(C^{\dotr}_\infty(\g,\h;\mathfrak{M})\big) \subseteq C^{\dotr}_\infty(\g,\h;\mathfrak{M})$.

\begin{definition}
Let \[C^{m}_\infty(\g,\h;\mathfrak{M}) \defeq \{w\in\mathfrak{M} \medotimes \Lambda^m_\infty\ |\ \iota(x)w = \theta(x)w= 0\quad \forall x \in \h \}.\]
The \subdefhighlight{subcomplex relative to $\h$} is the complex $\{C^{\dotr}_\infty(\g,\h;\mathfrak{M}),d\}$
\[
\begin{tikzcd}
    \dots \arrow[r, "d"] & C^{m-1}_\infty(\g,\h;\mathfrak{M}) \arrow[r, "d"] & C^{m}_\infty(\g,\h;\mathfrak{M}) \arrow[r, "d"] & C^{m+1}_\infty(\g,\h;\mathfrak{M}) \arrow[r, "d"] &\dots
\end{tikzcd}
\]
\end{definition}
The cohomology of this relative subcomplex is denoted $H^{\dotr}(\g,\h;\mathfrak{M})$.

\begin{lemma}
\label{lem: rel. SIC requirement}
When $\h=\mathfrak{z}$, the centre (or a central subalgebra) of $\g$, acts on $\M$ by scalars, $H^{\dotr}_\infty(\g,\z;\M)$ is non-trivial only if
$$\pi(z)= -\rho(z) = -\langle \beta, z \rangle, \quad \forall z\in\mathfrak{z}.$$
\end{lemma}
In this paper, when $\g$ is the symmetry algebra of an extended CFT, we call $H^{\dotr}(\g,\z;\M)$ the \subdefhighlight{BRST cohomology} of the $\g$ field theory, where $\pi\colon\g\to\M$ obeys lemma \ref{lem: rel. SIC requirement} and $\M$ is referred to as the \subdefhighlight{matter sector} of the $\g$ field theory. Note that lemma \ref{lem: rel. SIC requirement} is equivalent to that of anomaly cancellation setting the central charge of the matter sector, or the critical dimension, of string theories (e.g. when $\g=\Vir$). This will be made clearer in the following subsection dedicated to working through the construction of the semi-infinite cohomology of $\g_\lambda$.

\subsection{Example: the $\g_\lambda$ Lie algebra}

Let $\g_\lambda =\medoplus_{n\in\ZZ} (\g_\lambda)_n$, where $(\g_\lambda)_n = \CC L_n \oplus \CC M_n$ and likewise for the restricted dual $\g'_\lambda$. We choose the ordered basis
\begin{equation}
\label{eq: g_lambda basis e_i}
    e_{2i-1} = L_i, \quad e_{2i-2} = M_i.
\end{equation}
We now build the representation $\rho\colon\g_\lambda \to \End(\semiinfforms)$ according to \eqref{eq: rho_semiinfforms} using the basis \eqref{eq: g_lambda basis e_i}, $\beta=0$, and the normal-ordering prescription dictated by the vacuum semi-infinite form
\begin{equation}
\label{eq: g_lambda vacuum}
    \omega_0 = \omega_0 =  e'_{1} \wedge e'_{0} \wedge e'_{-1} \wedge \dots =  L'_{1} \wedge M'_{1} \wedge L'_0 \wedge \dots .
\end{equation}
\begin{lemma}
\label{lem: rho(L_n) and rho(M_n) semiinfforms}
Using the relations
\begin{alignat*}{2}
    &\ad_{L_n} L_i = (n-i) L_{n+i},\quad && \ad_{L_n} M_i = -(i+\lambda n) M_{n+i}\\
    &\ad_{M_n} L_i = (n+\lambda i) M_{n+i},\quad &&\ad_{M_n} M_i = 0,
\end{alignat*}
we have the following:
\begin{align}
    \rho(L_n) &= \sum_{i\in\ZZ} (n-i) :\iota(L_{i+n}) \varepsilon(L'_i): - \sum_{i\in\ZZ} (i+\lambda n) :\iota(M_{i+n})\varepsilon (M'_i):\label{eq: rho(L_n) g_lambda}\\
    \rho(M_n) &= \sum_{i\in\ZZ} (n + \lambda i) :\iota(M_{i+n}) \varepsilon(L'_i):\label{eq: rho(M_n) g_lambda}. 
\end{align}
\end{lemma}
\begin{corollary}
\label{cor: rho(L_n>0, M_n>0) omega_0 = 0}
    For the choice $\beta=0$, $\rho(L_{n\geq 0})\omega_0 = \rho(M_{n\geq0})\omega_0 = 0$.
\end{corollary}
With \eqref{eq: rho(L_n) g_lambda} and \eqref{eq: rho(M_n) g_lambda}
at hand, we proceed to compute the failure of
$\rho\colon\g_\lambda\to\End(\semiinfforms)$ to be a
representation. That is, we compute the 2-cocycle in proposition
\ref{prop: FGZ_prop_1.1} and check whether it can be made identically zero by an
appropriate choice of $\beta$. If this cannot be done, then
$\semiinfforms$ is at best a representation of a central extension of
${\g}_\lambda$ by that 2-cocycle. This simply requires the computations
$[\rho(L_n), \rho(L_{-n})]-\rho([L_n,L_{-n}])$ and
$[\rho(L_n), \rho(M_{-n})]-\rho([L_n,M_{-n}])$ acting on $\omega_0$,
since this is the only way we get something non-trivial. Doing so, we
observe that part of this failure is proportional to the Gelfand-Fuks
cocycle \ref{eq: Gelfand-Fuks cocycle}, and therefore is not a
cohomologically non-trivial contribution that can be absorbed by a
different choice of $\beta$. Thus, we have the following theorem (see
appendix \ref{app: proofs and calculations} for a proof).

\begin{theorem}
\label{thm: semiinfrep of g_lambda central charge}
    The space of semi-infinite forms on $\g_\lambda$, $\semiinfforms$, is a representation of $\Hat{\g}_\lambda$, where $\hat{\g}_\lambda$ is the central extension of $\g_\lambda$ by the Virasoro cocycle. The central element acts on $\semiinfforms$ as $\rho(c_L) = -\big(26 + 2(6\lambda^2-6\lambda+1)\big) \Id_{\semiinfforms}$.
\end{theorem}

\begin{remark}
Recall from \ref{rem:other-lambda} that for $\lambda=-1,0,1$, there exist other possible central charges coming from other central extensions of these Lie algebras. We show that these must be zero for $\rho\colon\hat{\g}_\lambda\to\semiinfforms$ to be a well-defined representation. See appendix \ref{app: proofs and calculations} for more details. 
\end{remark}

\subsubsection{The semi-infinite wedge representation as $bc$-systems}

We now construct the field-theoretic formulation of $\rho\colon\Hat{\g}_\lambda\to\End(\semiinfforms)$. The main result is summarised in the following proposition.
\begin{proposition}
\label{prop: bc-system ff realisation of g_lambda}
    The $\hat{\g}$-module structure of the space of semi-infinite forms $\semiinfforms$ of $\hat{\g}_\lambda$ is equivalent to a free field realisation of a $\hat{\g}_\lambda$ field theory in terms of two independent $bc$-systems of weights $(2,-1)$ and $(1-\lambda,\lambda)$, generated by $(b,c)$ and $(B,C)$ respectively. The resulting field theory has $\hat{\g}_\lambda$ symmetry generated by the fields
    \begin{equation}
        \left( \Tgh, \Mgh \right) \defeq \big(T^{bc}+T^{BC},\,(\lambda-1)\normalord{B\partial c} - \normalord{\partial B c}\big),
    \end{equation}
    where $T^{bc}$ and $T^{BC}$ are given by \eqref{eq: Tbc general}.
\end{proposition} 
\begin{proof}
We start with the space of semi-infinite forms and make contact with $bc$-systems as follows. Let
\begin{equation}
\label{eq: g_lambda ghost modes identification}
    b_n \defeq \iota(L_n), \quad c_n \defeq \varepsilon(L'_{-n}), \quad B_n \defeq \iota(M_n), \quad C_n \defeq  \varepsilon(M'_{-n})
\end{equation}
and define the generating functions
\begin{gather}
    b(z) \defeq \sum_{n\in\ZZ}b_n z^{-n-2} \quad\quad c(z) \defeq \sum_{n\in\ZZ} c_n z^{-n+1} \label{eq: b,c gen func g_lambda}\\
    B(z) \defeq \sum_{n\in\ZZ}B_n z^{-n-(1-\lambda)} \quad\quad C(z) \defeq \sum_{n\in\ZZ} C_n z^{-n-\lambda} \label{eq: B,C gen func g_lambda}.
\end{gather}
These fields, together with their construction of their respective Virasoro elements as described in definition \ref{def: bc-systems general}, satisfy the properties of weight $(2,-1)$ and weight $(1-\lambda,\lambda)$ fermionic $bc$-systems respectively.

The key principle is to construct generating functions
\begin{equation}
    \Tgh(z) = \sum_{n\in\ZZ}\rho(L_n)z^{-n-2} \quad \text{ and } \quad M^\text{gh}(z) = \sum_{n\in\ZZ}\rho(M_n)z^{-n-(1-\lambda)}
\end{equation}
from the fields $b(z)$, $c(z)$, $B(z)$ and $C(z)$. This can be done by looking at $\rho(L_n)$ and $\rho(M_n)$ in more detail. Using  \eqref{eq: g_lambda ghost modes identification} in lemma \ref{lem: rho(L_n) and rho(M_n) semiinfforms},
\begin{align}
    \rho(L_n) &= \sum_{m\in\ZZ} (n-m) :\big(b_{m+n} c_{-n} + B_{m+n} C_{-n}\big):\label{eq: rho(L_n) g_lambda in ghost modes}\\
    \rho(M_n) &= \sum_{m\in\ZZ} (n + \lambda m) :B_{m+n} c_{-n}:\label{eq: rho(M_n) g_lambda in ghost modes}.
\end{align}
Hence, we may ask: what normal-ordered products of $b(z)$, $c(z)$, $B(z)$ and $C(z)$ have modes \eqref{eq: rho(L_n) g_lambda in ghost modes} and \eqref{eq: rho(M_n) g_lambda in ghost modes}? The answer is straightforward for $\Tgh$:
\begin{equation}
\label{eq: g_lambda Tgh}
    \Tgh = -2\normalord{b\partial c} - \normalord{\partial b c} -(1-\lambda)\normalord{B\partial C}+\lambda \normalord{\partial B C} = T^{bc} + T^{BC},
\end{equation}
where $T^{bc}$ and $T^{BC}$ are given by \ref{eq: Tbc general}. Thus, the form of $\Tgh$ is exactly what one would expect when considering the total Virasoro element of two independent $bc$-systems. 

The answer to the earlier question is not as obvious for $M^\text{gh}$, but the form of $\Tgh$ is quite instructive in helping us guess what terms should be there. $\rho(L_n)$ has one term with $b$ and $c$ modes and another with $B$ and $C$ modes. The corresponding field $\Tgh$, whose $n$-th mode is $\rho(L_n)$, is a linear combination of weight 2 fields formed from the normal-ordered products of one $b$ and one $c$, and one $B$ and one $C$. Now consider $\rho(M_n)$. Since only $B$ and $c$ modes appear in \eqref{eq: rho(M_n) g_lambda in ghost modes}, we infer, based on the form of $\Tgh$ in relation to $\rho(L_n)$, that the most general expression for the corresponding field $M^\text{gh}$ is
\begin{equation}
    M^\text{gh} = a_1 \normalord{B\partial c} + a_2\normalord{\partial B c},
\end{equation}
for some $a_1,a_2\in\CC$. A quick computation reveals that $a_1=\lambda-1$ and $a_2=-1$. Thus,
\begin{equation}
\label{eq: g_lambda Mgh}
    M^\text{gh} = (\lambda-1) \normalord{B\partial c}- \normalord{\partial B c}.
\end{equation}
Now that we have our fields $\Tgh$ and $M^\text{gh}$, we compute their
OPEs (using Mathematica~\cite{Thielemans:1991uw,Thielemans:1994er})
and arrive at equations \eqref{eq: g_lambda TT OPE}, \eqref{eq:
  g_lambda TM OPE} and \eqref{eq: g_lambda MM OPE} with
$c_L = -26-2(6\lambda^2-6\lambda+1)$, as expected. This shows that our
field-theoretic formulation of the semi-infinite wedge representation
of $\Hat{\g}_\lambda$ is consistent and thereby completes the proof.
\end{proof}

\subsubsection{The BRST quantisation of $\Hat{\g}_\lambda$ field  theories}

We are now ready to explain the BRST quantisation of
$\Hat{\g}_\lambda$ field theories in the language of semi-infinite
cohomology. This quantisation procedure requires the construction of a
square-zero BRST operator, which is precisely the semi-infinite
differential \eqref{eq: semi-infinite differential}. The resulting
BRST cohomology with respect to this operator is then
$H^{\dotr}_\infty(\hat{\g}_\lambda,c_L;\M)$. According to lemma
\ref{lem: rel. SIC requirement}, $\pi\colon\Hat{\g}_\lambda \to\End\M$
needs to be a category $\mathcal{O}$ representation with
$\pi(c_L)=-\rho(c_L)=26 + 2(6\lambda^2-6\lambda+1)$ to ensure that the
BRST operator is square-zero. The field-theoretic formulation is
complete when $\M$ is regarded as the \emph{matter sector} of the
$\hat{\g}_\lambda$ field theory generated by
\begin{equation}
    \label{eq: Tmat Mmat construction}
    \Tmat(z) = \sum_{n\in\ZZ}\pi(L_n)z^{-n-2}, \quad \quad \Mmat(z) = \sum_{n\in\ZZ}\pi(M_n)z^{-n-(1-\lambda)}.
\end{equation}
These obey \eqref{eq: g_lambda TT OPE}, \eqref{eq: g_lambda TM OPE} and \eqref{eq: g_lambda MM OPE}.
$\Tmat$ is the Virasoro element of this $\hat{\g}_\lambda$ field theory with central charge $\pi(c_L)$ that cancels that of the \emph{ghost sector} of the theory, which is precisely the semi-infinite wedge representation $\rho\colon\hat{\g}_\lambda\to\End(\semiinfforms)$. 

We summarise the construction of the BRST operator in the following theorem.
\begin{theorem}
\label{thm: BRST quantisation of ghat_lambda field theories}
    Let $(\Tmat, \Mmat)$ as constructed via \eqref{eq: Tmat Mmat construction} generate a $\Hat{\g}_\lambda$ field theory with central charge $c^\text{mat}$ and $(T^\text{gh}, M^\text{gh})$ be the fields \eqref{eq: g_lambda Tgh} and \eqref{eq: g_lambda Mgh}. The zero mode $d$ of the BRST current 
    \begin{equation}
    \label{eq: BRST current}
        \mathfrak{J} = \normalord{c\Tmat} + \frac{1}{2}\normalord{cT^\text{gh}} + \normalord{C\Mmat} + \frac{1}{2}\normalord{CM^\text{gh}},
    \end{equation}
    otherwise known as the BRST operator, is square-zero if and only if $c^\text{mat} = 26+2(6\lambda^2-6\lambda+1)$. This is the field-theoretic restatement of lemma \ref{lem: rel. SIC requirement}.
    Define $\Ttot \defeq \Tmat+T^\text{gh}$ and
    $\Mtot \defeq \Mmat + M^\text{gh}$. Then
    \begin{equation}
    \label{eq: T^tot and M^tot are BRST-exact}
        db = \Ttot, \quad \quad dB = \Mtot.
    \end{equation}
    This is the field-theoretic restatement of the fact that the semi-infinite differential \eqref{eq: semi-infinite differential} obeys \eqref{eq: theta_as_anticomm}.
\end{theorem}
\begin{corollary}
    Equation \eqref{eq: T^tot and M^tot are BRST-exact} along with the fact that $L^\text{tot}_0$ acts semi-simply on $\M\medotimes\semiinfforms$ implies that all non-trivial BRST cohomology resides only in the zero-eigenvalue eigenspace of $L^\text{tot}_0$.
\end{corollary}

\section{Embedding theorems}\label{sec: embedding theorems}

The BRST cohomology of a TCFT has more structure than just that of a
graded vector space. It is actually a Batalin--Vilkovisky algebra (see
\cite{Lian:1992mn, Lian:1994na, Lian:1995wa}).
Therefore we need to
be precise when we talk about isomorphisms of BRST cohomology. In
this section we will exhibit some BV-algebra isomorphisms of BRST
cohomologies. These are ``stronger'' than vector space
isomorphisms. The main idea is to show that these isomorphisms
preserve the extra structure (i.e. that of a BV algebra) manifestly,
without reference to the exact details of the structure. In this
section, we present a couple of embedding theorems relating the BRST
cohomology of the Virasoro, $\Hat{\g}_\lambda$ and the twisted $N=2$
superconformal algebras.

\subsection{Embedding 1: $c_L=26$ CFTs into $\hat{\g}_\lambda$ field theories}

The first embedding theorem can be stated as follows.

\begin{theorem}\label{thm: Embedding Vir with c=26 into g_lambda}
Let $\M$ be a $\Vir$-module with central charge 26. The BRST cohomology of a CFT with matter sector $\M$ (i.e., generated by $T^{\M}$) is isomorphic, as a BV algebra, to the BRST cohomology of a $\hat{\g}_\lambda$ field theory with $V^{\beta\gamma}_q \otimes \M$ as its matter sector, where the $\Hat{\g}_\lambda$-module $V^{\beta\gamma}_q$ is the space of states of the $\beta\gamma$-system with any vacuum choice $\ket{q}$.
More succinctly, 
 \begin{equation}
   H^{\dotr}_\infty (\Vir, c_L; \M) \cong H^{\dotr}_\infty (\hat{\g}_\lambda, c_L; V^{\beta\gamma}_q \otimes \M)
\end{equation}   
as BV algebras.
\end{theorem}
To prove this, we need two ingredients. Evidenced by its appearance in the theorem statement, the first is the free field realisation of $\hat{\g}_\lambda$ in terms of an appropriately weighted $\beta\gamma$-system. This can be obtained by a straightforward generalisation of the same construction for the BMS\textsubscript{3} algebra which was done in \cite{Banerjee:2015kcx}.
\begin{lemma}
\label{lem: g_lambda FFR with beta-gamma}
There exists a free field realisation of every $\Hat{g}_\lambda$ field theory in terms of a weight $(1-\lambda,\lambda)$ $\beta\gamma$-system given by
\begin{equation}\label{eq: g_lambda beta-gamma embedding}
    (T,M) \rightarrow (T^{\beta\gamma},\beta),
\end{equation}
where $T^{\beta\gamma}$ is constructed according to definition \ref{def: bc-systems general} and has central charge $c_L = 2(6\lambda^2-6\lambda+1)$.
\end{lemma}
\begin{proof}
    Computing the OPEs of $(T^{\beta\gamma}, \beta)$ using the properties given by \ref{item D4} shows that the embedding \eqref{eq: g_lambda beta-gamma embedding} indeed satisfies \eqref{eq: g_lambda TT OPE}, \eqref{eq: g_lambda TM OPE} and \eqref{eq: g_lambda MM OPE}, with the Virasoro central charge $c_L = 2(6\lambda^2-6\lambda+1)$.
\end{proof}

The second ingredient is the Koszul CFT.
\begin{definition}
\label{def: Koszul CFT}
    A \subdefhighlight{Koszul CFT} is spanned by a $\Tilde{\beta}\Tilde{\gamma}$-system and a $\Tilde{b}\Tilde{c}$-system, each of weight $(1-\mu,\mu)$. Together with the differential $d_\text{KO} \defeq \normalord{\Tilde{c}\Tilde{\beta}}_0$ and some choice of vacuum $\ket{q}$, a Koszul CFT describes a differential graded algebra spanned by the modes $\{\Tilde{\beta}_n,\Tilde{\gamma}_n,\Tilde{b}_n,\Tilde{c}_n\}_{n\in\ZZ}$. The cohomology of the differential graded algebra described by any Koszul CFT with respect to the differential $d_\text{KO}$, denoted $H^{\dotr}_\text{KO}$, is called the \subdefhighlight{chiral ring} of a Koszul CFT.
\end{definition}
\begin{lemma}
    \label{lem: Koszul chiral ring}
    The chiral ring of a Koszul CFT is 1-dimensional. That is,
    \begin{equation}
        \label{eq: Koszul CFT cohomology}
        H^{n}_\text{KO}=\begin{cases}
             \CC\ket{\text{vac}}_q, \quad n=0 \\
             0, \quad \text{otherwise,}
        \end{cases}
    \end{equation}
    where $\ket{\text{vac}}_q \defeq \ket{-q}_{\beTil\gamTil}\otimes \ket{q}_{\Tilde{b}\Tilde{c}}$.
\end{lemma}
\begin{proof}
    This follows from the Kugo-Ojima quartet mechanism (see \cite{Kato:1982im}). 
\end{proof}

We are now ready to prove theorem \ref{thm: Embedding Vir with c=26 into g_lambda}.

\begin{proof}[Proof of Theorem \ref{thm: Embedding Vir with c=26 into g_lambda}]
We construct an explicit inner automorphism of the Lie superalgebra structure on $\End (V^{\beta\gamma}\medotimes \M)$ which splits $d$, the BRST operator of $\hat{\g}_\lambda$ (zero mode of \eqref{eq: BRST current}), into that of the Virasoro CFT, $d_{\Vir}$, and a Koszul differential $d_\text{KO}$. This is analogous to Ishikawa and Kato's proof for the embedding of the bosonic string into the N=1 superstring \cite{IshikawaKato1994}.

Let $(T^{\M}+ T^{\beta\gamma}, \beta)$ describe the matter sector of the $\hat{\g}_\lambda$ field theory, where the $\beta\gamma$-system is of weight $(1-\lambda)$ (see lemma \ref{lem: g_lambda FFR with beta-gamma}). Let $(T^\text{gh}, M^\text{gh})$ given by \eqref{eq: g_lambda Tgh} and \eqref{eq: g_lambda Mgh} describe the ghost sector. 
Starting with \eqref{eq: BRST current} and using $\Tmat = T^{\beta\gamma}+T^{\M}$ and \eqref{eq: g_lambda Tgh},
\begin{equation}
\begin{split}
    \mathfrak{J} &=  \normalord{cT^{\M}}+\normalord{cT^{\beta\gamma}}+\frac{1}{2}\normalord{cT^{bc}}+\frac{1}{2}\normalord{cT^{BC}}+\normalord{C\beta}+\frac{1}{2}\normalord{C M^\text{gh}}\\
    &=\mathfrak{J}_{\Vir} + \mathfrak{J}_\text{KO} + \normalord{cT^{\beta\gamma}}+\frac{1}{2}\normalord{cT^{BC}}+\frac{1}{2}\normalord{C M^\text{gh}},
\end{split}
\end{equation}
where we have defined
\begin{align}
    \mathfrak{J}_{\Vir} &= \normalord{cT^{\M}}+\frac{1}{2}\normalord{cT^{bc}}\\
    \mathfrak{J}_\text{KO} &= \normalord{C\beta}.
\end{align}
The BRST operator of the Virasoro CFT, $d_{\Vir}$, the Koszul differential, $d_\text{KO}$, and the BRST operator of the $\hat{\g}_\lambda$ field theory, $d$, are the zero modes of $\mathfrak{J}_{\Vir}$, $\mathfrak{J}_\text{KO}$ and $\mathfrak{J}$ respectively.

Next, we construct the weight 1 bosonic field
\begin{equation}
    r(z)= \sum_{n\in\ZZ} r_n z^{-n-1} = (\lambda-1)\normalord{\gamma B \partial c}-\normalord{\gamma \partial B c},
\end{equation}
whose zero mode generates the similarity transformation which performs the splitting of $d$ as intended: 
\begin{equation}
     d = \exp({\ad_{r_0}})(d_\text{KO} + d_{\Vir})= \exp(r_0)(d_\text{KO} + d_{\Vir})\exp(-r_0).
\end{equation}

Thus, by the K\"unneth Formula, 
\begin{equation}
    H^{\dotr}_\infty (\hat{\g}_\lambda, c_L; V^{\beta\gamma} \otimes \M) \cong H^{\dotr}_\infty (\Vir, c_L; \M) \medotimes H^{\dotr}_\text{KO}
\end{equation}
as BV algebras. Finally, lemma \ref{lem: Koszul chiral ring} implies
that the only linearly independent state of the weight
$(1-\lambda,\lambda)$ $\beta\gamma$ and $BC$-systems which is
non-trivial in $d_\text{KO}$-cohomology is the choice of vacuum. This
completes the proof.
\end{proof}

The statement of theorem \ref{thm: Embedding Vir with c=26 into
  g_lambda} is a result of embedding the semi-infinite complex
of the Virasoro algebra with values in some $c=26$ $\Vir$-module,
$\M$, into the relative semi-infinite  complex\footnote{We may
  replace `relative semi-infinite complex' with `BRST complex'} of the
$\hat{\g}_\lambda$ algebra with values in the $\hat{\g}_\lambda$
module $V^{\beta\gamma}\otimes \M$, where $\{M_n\}_{n\in\ZZ}$ act
trivially on $\M$. The embedding is therefore specific to a choice of
$\hat{\g}_\lambda$-module. Indeed, this is what the field-theoretic
formulations describes too.

On the other hand, the next embedding theorem is not an embedding of
complexes but rather a construction of a twisted $N=2$ superconformal
algebra using the modes of the complex formed from tensoring the
semi-infinite complex of $\hat{\g}_\lambda$ with the Koszul
complex. It holds for any choice of $\hat\g_\lambda$-module with
central charge such that the BRST operator is square-zero.

\subsection{Embedding 2: twisted $N=2$ SCFTs from $\hat{\g}_\lambda$ field theories}

The theorem we present in this subsection arose from testing the
conjecture made in \cite{Figueroa-OFarrill:1995agp} and then refined
in \cite{Figueroa-OFarrill:1995qkv}. This conjecture states that every
topological conformal field theory (TCFT) is homotopy equivalent to a
twisted N=2 SCFT \cite{Figueroa-OFarrill:1995qkv}. It originated from a search for a ``universal string theory'' \cite{Berkovits:1993xq}.
We refer the
reader to \cite{Figueroa_O_Farrill_1993, Figueroa-OFarrill:1995qkv, Getzler_1994} for a definition and/or review of TCFTs and N=2 SCFTs.

Throughout this subsection, let
$(\Tilde{b},\ \Tilde{c},\ \Tilde{\beta},\ \Tilde{\gamma})$ form a
Koszul CFT as described in definition \ref{def: Koszul CFT}. Let
$(b,\ c,\ B,\ C)$ form the ghost sector of a $\hat{g}_\lambda$ field
theory as described by expressions for $T^\text{gh}$ and $M^\text{gh}$
in \eqref{eq: g_lambda Tgh} and \eqref{eq: g_lambda Mgh}, and let
$(T^\M, M^\M)$ generate the matter sector with central charge
$26+2(6\lambda^2-6\lambda+1)$. We may then state the theorem as
follows.

\begin{theorem}\label{thm: Twisted N=2 SCFT}
    The BRST cohomology of a $\hat{\g}_\lambda$ field theory given by
    $(T^\M, M^\M)$ is isomorphic as a BV algebra to the chiral ring of
    a twisted $N=2$ SCFT. In other words, for any
    $\hat{\g}_\lambda$-module with $c_L =
    26+2(6\lambda^2-6\lambda+1)$, there exists a twisted $N=2$ SCFT
    whose chiral ring is isomorphic to the semi-infinite
    cohomology $H^{\dotr}_\infty(\hat{\g}_\lambda, c_L; \M)$ of
    $\hat{\g}_\lambda$ relative the centre.
\end{theorem}

The proof simply involves constructing the TCFTs corresponding to
$\hat{\g}_\lambda$ field theories and Koszul CFTs (presented as the
following two lemmas), taking their tensor product, and constructing a twisted
$N=2$ SCFT whose chiral ring is the tensor product of that of the
$\hat{\g_\lambda}$ and Koszul TCFTs. Doing so shows that the conjecture
in \cite{Figueroa-OFarrill:1995qkv} holds true for all
$\hat{\g}_\lambda$ field theories. 

\begin{lemma}
\label{lem: g_lambda TCFT}
    The TCFT given by the fields
    \begin{equation}
        \label{eq: g_lambda TCFT}
        \begin{split}
            \mathbb{G}^+ &= \normalord{cT^\M} + \frac{1}{2}\normalord{cT^\text{gh}} + \normalord{CM^\M} + \frac{1}{2}\normalord{CM^\text{gh}} \\
            \mathbb{G}^- &= b \\
            \mathbb{T} &= T^{\M} + T^\text{gh} \\
            \mathbb{J} &= -\normalord{bc}-\normalord{BC}.
        \end{split}
    \end{equation}
    describe a $\hat{\g}_\lambda$ field theory with matter sector $\M$. Its BRST cohomology is now taken with respect to the operator $\mathcal{Q}\defeq [\mathbb{G}^+,-]_1$.
  \end{lemma}

\begin{lemma}
    The TCFT given by the fields
    \begin{equation}
        \label{eq: Koszul TCFT}
        \begin{split}
            \mathbb{G}^+_K &= \normalord{\Tilde{c}\Tilde{\beta}}\\
            \mathbb{G}^-_K &= w\normalord{\Tilde{b}\partial\Tilde{\gamma}}-(1-w)\normalord{\partial \Tilde{b} \Tilde{\gamma}} \\
            \mathbb{J}_K &= -w\normalord{\Tilde{b}\Tilde{c}}-(1-w)\normalord{\Tilde{\beta}\Tilde{\gamma}} \\
            \mathbb{T}_K &= -w\normalord{\Tilde{b}\partial \Tilde{c}-\Tilde{\beta}\partial\Tilde{\gamma}} + (1-w)\normalord{\partial \Tilde{b} \Tilde{c}-\partial\Tilde{\beta}\Tilde{\gamma}}.
        \end{split}
    \end{equation}
    is a twisted N=2 SCFT description of a Koszul CFT. The cohomology is taken with respect to the differential $\mathcal{Q}_K\defeq [\mathbb{G}^+_K,-]_1 = d_\text{KO}$.
\end{lemma}
It is worth reiterating that the above lemmas do not present any new information about the $\hat{\g}_\lambda$ field theory and Koszul CFT; they are simply repackagings of the existing data of the field theories. Equipped with these lemmas, we are ready to present the proof.

\begin{proof}[Proof of Theorem \ref{thm: Twisted N=2 SCFT}]
    On its own, the $\hat{\g}_\lambda$ TCFT in lemma \ref{lem: g_lambda TCFT} cannot be modified using the available fields to obtain a twisted N=2 SCFT. However, this becomes possible once we tensor the $\hat{\g}_\lambda$ TCFT with a Koszul TCFT. From the field content of this larger TCFT, we may assemble
    \begin{equation}
        \label{eq: g_lambda N=2 embedding}
        \begin{split}
            \mathbb{G}^{+}_{N=2} &= \mathbb{G}^+ + \mathbb{G}^+_K +\partial X\\
            \mathbb{G}^{-}_{N=2} &= \mathbb{G}^- + \mathbb{G}^-_K \\
            \mathbb{T}_{N=2} &= \mathbb{T} + \mathbb{T}_K \\
            \mathbb{J}_{N=2} &= \mathbb{J} + \mathbb{J}_K + (1-\lambda)\normalord{\beTil\gamTil-\bTil\cTil}-(1-\lambda)\partial\normalord{c\bTil\gamTil}, 
        \end{split}
    \end{equation}
    where
    \begin{equation}
        X = \tfrac{1}{2}(1+\lambda)\normalord{cBC}
        -(1-\lambda)\normalord{c\normalord{\Tilde{\beta}\Tilde{\gamma}-\Tilde{b}\Tilde{c}}+ \partial c c \bTil \gamTil} + (2-\lambda) \partial c.
    \end{equation}
    By computing OPEs, we can check that these fields indeed describe a twisted $N=2$ SCFT. 
    The chiral ring of any twisted $N=2$ SCFT is the cohomology taken with respect to the differential $\mathcal{Q}_{N=2}\defeq[\mathbb{G}^+_{N=2},-]$.
    In this case, due to \ref{item: proporties of partial in CFT}, the total derivative term in $\mathbb{G}^+_{N=2}$ does not contribute to $\mathcal{Q}_{N=2}$. Thus $\mathcal{Q}_{N=2} = \mathcal{Q}+\mathcal{Q}_K$ indeed. $\mathbb{J}^+_{N=2}$ is also a sum of $\mathbb{J}$ and $\mathbb{J}_K$, up to extra terms that are trivial in $\mathcal{Q}_{N=2}$-cohomology due to lemma \ref{lem: Koszul chiral ring}. Hence, all 4 fields of the twisted $N=2$ SCFT are, up to cohomologically trivial terms, equal to the corresponding fields in the tensor product of the $\hat{\g}_\lambda$ and Koszul TCFTs. Thus, by the K\"unneth formula, the chiral ring of the twisted $N=2$ SCFT given by \eqref{eq: g_lambda N=2 embedding} is isomorphic to the tensor product of the BRST cohomology of the $\hat{\g}_\lambda$ field theory given by $(T^\M, M^\M)$ and the chiral ring of the Koszul CFT. Since the latter is spanned by just its vacuum state, this completes the proof.
\end{proof}

\section{Case $\lambda=-1$: The BMS\textsubscript{3} Lie algebra}\label{sec: BMS3}

The universal central extension of $\g_\lambda$ when $\lambda=-1$ is
isomorphic to the BMS\textsubscript{3} algebra. This algebra was first introduced to the Lie algebra and VOA literature by Zhang and Dong \cite{zhang2007walgebra} as the $W(2,2)$ algebra. Since the BMS\textsubscript{3} algebra is the symmetry algebra of the closed tensionless string \cite{Isberg:1993av, Bagchi:2020fpr}, $\lambda=-1$ is an interesting case to explore in more detail. However, the tensionless string does not admit a holomorphically factorisable field-theoretic description, so the BMS\textsubscript{3} algebra does not appear as the symmetry algebra of the chiral part of some full CFT. This is contrary to how the $\hat\g_\lambda$ algebra appears in our field theoretic descriptions. Nonetheless, the results we present are intrinsic to the Lie algebra, and not the field theory it describes. The field-theoretic formulation only serves as a computational tool for the construction of the semi-infinite cohomology of the Lie algebra and the results that follow. How one wishes to extrapolate these findings on the BMS\textsubscript{3} algebra to BMS\textsubscript{3} field theories is a separate matter, on which we shed some light in the last section.

\subsection{No BRST operator for $c_M \neq 0$?}
Let us remind ourselves that the BMS\textsubscript{3} algebra is the vector space $\bigoplus_{n\in\ZZ} \CC L_n \medoplus \CC M_n$ with Lie bracket
\begin{equation}
\begin{split}
    [L_n, L_m] &= (n-m) L_{m+n} + \frac{1}{12} n(n^2-1) \delta^0_{m+n} c_L\\
    [L_n, M_m] &= (n-m) M_{m+n} +\frac{1}{12} n(n^2-1) \delta^0_{m+n} c_M \\
    [M_n, M_m] &= 0.
\end{split}
\end{equation}

Now consider a $\ZZ$-graded BMS\textsubscript{3}-module $\M$ with central charges denoted $(c_L, c_M)$. An important consequence of theorem \ref{thm: semiinfrep of g_lambda central charge} is that there does \emph{not} exist a BRST operator for the BMS\textsubscript{3} algebra when $c_M\neq 0$. The calculation performed for generic $\lambda$ to construct $\rho\colon\g_\lambda\rightarrow\End\semiinfforms$ proves that this must be the case. A closer look at the calculation (see appendix \ref{app: proofs and calculations}) shows that this is feature is specifically due to $\{M_n\}_{n\in\ZZ}$ forming an abelian ideal of $\g_\lambda$, corroborating the fact that the central extension which one needs to use is always the Virasoro one for any value of $\lambda$, including the cases $\lambda=-1,0,1$ where other central extensions exist. Nonetheless, for $\lambda=-1$, we present another argument as to why this must be the case by going beyond the construction of the semi-infinite wedge representation of ${\g}_\lambda$.

Consider a BMS\textsubscript{3} field theory generated by some $T$ and $M$. Purely from the perspective of gauge theory, we would need to introduce two sets of ghosts - one for $T$ and the other for $M$, of appropriate weights, to gauge the BMS\textsubscript{3} symmetry of the theory. These are the weight $(2,-1)$ $bc$-system and weight $(1-\lambda,\lambda)$ $BC$-system respectively, where the latter is also of weight $(2,-1)$ for $\lambda=-1$. These should themselves assemble into some $T^{gh}$ and $M^{gh}$ which generate BMS\textsubscript{3} symmetry via their OPEs.
Proposition \ref{prop: bc-system ff realisation of g_lambda} already tells us how to do this for $(c_L, c_M) = (-52,0)$. However, we now step away from semi-infinite wedge representations and consider, more generally, any normal-ordered products of the fields $b$, $c$, $B$ and $C$ to obtain a bosonic weight 2 field which is quasiprimary with respect to $T^{gh} = T^{bc}+T^{BC}$ as given in \eqref{eq: Tbc general}.
The table below summarises every possible weight 2 bosonic term that one could form from the fields of the two $bc$-systems. 
\begin{center}
    \begin{tabular}{c|c}
      No. of $b,c,B,C$   & Term \\ \hline
    1 & None  \\
    2 & $\normalord{b\partial c}$, $\normalord{\partial b c}$, $\normalord{b \partial C}$, $\normalord{\partial b C}$, $\normalord{B\partial c}$, $\normalord{\partial B c}$, $\normalord{B\partial C}$, $\normalord{\partial B C}$ \\
    3 & None \\
    4 & $\normalord{b c B C}$ \\
    \end{tabular}
\end{center}
Terms with 5 or more $b,c,B,C$ that are bosonic and weight 2 will necessarily vanish. Taking the most linear combination of these fields
\begin{equation}
\begin{split}
   M^\text{gh} = &\alpha_1\normalord{b\partial c} + \alpha_2\normalord{\partial b c} + \alpha_3\normalord{b \partial C} + \alpha_4\normalord{\partial b C} + \alpha_5\normalord{B\partial c} + \alpha_6\normalord{\partial B c} \\ 
   + &\alpha_7\normalord{B\partial C} + \alpha_8\normalord{\partial B C} + \alpha_9 \normalord{b c B C} 
\end{split}
\end{equation}
and enforcing the OPEs 
\begin{equation}
    [T^\text{gh}M^\text{gh}]_4 = \kappa\mathbbm{1},\quad 
    [T^\text{gh}M^\text{gh}]_3 = 0,\quad
    [T^\text{gh}M^\text{gh}]_2 = 2M^{gh},\quad 
    [T^\text{gh}M^\text{gh}]_1= \partial M^{gh},
\end{equation}
we obtain the following result.

\begin{proposition}
    \label{prop: BMS3 with non-zero c_M from fermionic bc-systems}
    There exists a realisation of a BMS\textsubscript{3} field theory
    in terms of two weight $(2,-1)$ $bc$-systems with central charges
    $(-52,-c_M)$ for any nonzero value of $c_M$ given by
    \begin{equation}
    \label{eq: BMS3 with c_M non-zero}
        \begin{split}
            \Tgh &= -2\normalord{b\partial c} - \normalord{\partial b c} - 2\normalord{B\partial C} - \normalord{\partial B C} \\
            \Mgh &= - \tfrac{c_M}{54} \left( -\normalord{bcBC} + \normalord{b\partial C} + \normalord{\partial c B} + \partial\left(\frac{3}{2} \normalord{b c} + \frac{3}{2}\normalord{B C} + \normalord{b C} + \normalord{c B}\right)  \right).
        \end{split}
    \end{equation}
\end{proposition}

Proposition \ref{prop: BMS3 with non-zero c_M from fermionic bc-systems} shows that $c_M \neq 0$ is achieved through a term that is quartic in the fields of the $bc$-systems. 
It is the emergence of this term in the ghost sector of a BMS\textsubscript{3} field theory which causes the BRST operator to no longer be square-zero.

\begin{proposition}
  Let $(T,\,M)$ generate a BMS\textsubscript{3} field theory with
  $c_L=52$ and $c_M\neq 0$. Let $(\Tgh, \Mgh)$ be given by \eqref{eq:
    BMS3 with c_M non-zero}. Then the zero mode of the BRST current
    \begin{equation}
        j_\text{BRST} = \normalord{cT} + \frac{1}{2}\normalord{c\Tgh} + \normalord{CM} + \frac{1}{2}\normalord{C\Mgh}
    \end{equation}
    is not square-zero. Alternatively, there does not exist any BRST
    differential such that $db = M + \Mgh$.
\end{proposition}

\subsection{Physical realisations of chiral BMS\textsubscript{3} field theories}

We present string theories which can be studied as chiral BMS\textsubscript{3} field theories from the perspective of the worldsheet. An example would be the bosonic ambitwistor string in D-dimensional Minkowski spacetime \cite{MasonSkinner}. Its worldsheet description is given by D weight $(0,1)$ $\beta\gamma$-systems labelled by a spacetime index $\mu\in\{0,1,...,D-1\}$. The Virasoro element $T^{\text{amb}}$ is what we would expect, while the weight 2-primary $M^{\text{amb}}$ is constructed from the normal-ordered products of the weight 1 primaries $\gamma_\mu$:
\begin{align}
    T^{\text{amb}} &= -\normalord{\partial\beta^{\mu}\gamma_\mu}\label{eq: ambitwistor Vir element}\\
    M ^{\text{amb}} &= \eta^{\mu\nu}\normalord{\gamma_\mu\gamma_\nu} \label{eq: ambitwistor M}
\end{align} 
This is a BMS\textsubscript{3} field theory with central charge $(c_L, c_M) = (2D,0)$. Thus, the 26-dimensional ambitwistor string with $(c_L, c_M)= (52,0)$ admits a sensible BRST complex from which we can compute BRST cohomology. This is consistent with both the fact that its spectrum should emerge as the NR contraction of two copies of the Virasoro algebra \cite{Bagchi:2020fpr} and that its critical dimension is 26 \cite{MasonSkinner, Casali:2016atr, Bagchi:2020fpr}.

Another realisation one could consider comes from the Nappi--Witten string \cite{Nappi:1993ie}. Consider the complexification of the Nappi--Witten algebra for convenience, generated by $P^{\pm}$, $I$ and $J$. The Lie bracket on these generators is
\begin{equation}
   [P^+, P^-] = I, \quad [J, P^{\pm}] = \pm P^{\pm}
\end{equation}
 and the invariant inner product is
 \begin{equation}
    \langle P^+, P^-\rangle = 1, \quad \langle I,J \rangle = 1
 \end{equation}
 and zero otherwise.
Without any extra effort, we may consider a higher dimensional analogue (also considered in \cite{Sfetsos:1993na}) given by the Lie algebra $\mathfrak{nw}_{2n+2}$ generated by  $\{P_a^{\pm}\}_{a\in\{1,\dots,n\}}$, $I$ and $J$, with Lie bracket
\begin{equation}
   [P^+_a, P^-_b] = \delta_{ab} I, \quad [J, P^{\pm}_a] = \pm P^{\pm}_a,
\end{equation}
 and invariant inner product
 \begin{equation}
    \langle P^+_a, P^-_b\rangle = \delta_{ab}, \quad \langle I,J \rangle = 1
 \end{equation}
and zero otherwise. These translate into the OPEs of the corresponding currents
\begin{align}
    P^+_a (z) P^-_b (w) &= \frac{\delta_{ab}\mathbbm{1}(w)}{(z-w)^2} + \frac{\delta_{ab} I(w)}{z-w} + \reg. \label{eq: gen NW OPE for P+P-}\\
     J(z) P^{\pm}_a(w) &= \frac{\pm P^{\pm}_a(w)}{z-w} + \reg. \label{eq: gen NW OPE for J P+-} \\
     J(z) I(w) &= \frac{\mathbbm{1}(w)}{(z-w)^2} + \reg. \label{eq: gen NW OPE for JI}
\end{align}
As usual, the modes of each of the weight 1 fields
$P^{\pm}_1(z),\dots, P^{\pm}_{2n}(z), I(z), J(z)$ obey the
affinisation $\widehat{\mathfrak{nw}}_{2n+2}$ of
$\mathfrak{nw}_{2n+2}$. Hence, we may build a Virasoro element via the
Sugawara construction (as done in \cite{Nappi:1993ie} for $n=1$)
\begin{equation} \label{eq: Tsug BMS3 from NW gen}
  T^\text{sug} = \sum_{a=1}^n\normalord{P^+_a P^-_a} + \normalord{IJ} -\frac{n}{2}\partial I - \frac{n}{2}\normalord{II},
\end{equation}
with central charge $2n + 2$. 
Likewise, one can also construct a weight 2 primary 
\begin{equation} \label{eq: Msug BMS3 from NW gen}
  M^\text{sug} = \normalord{II}.
\end{equation}
$M^\text{sug}$ from weight one currents as well. Thus, $(T^\text{sug},
M^\text{sug})$ given by \eqref{eq: Tsug BMS3 from NW gen} and
\eqref{eq: Msug BMS3 from NW gen} give a realisation of a
BMS\textsubscript{3} field theory via the Sugawara construction
applied to the higher dimensional generalisation of the Nappi--Witten
algebra. This realisation has central charges $(c_L,c_M) = (2n+2,0)$. 
Hence, setting $n=25$ indeed gives a BMS\textsubscript{3} field theory
of central charge $c_L = 52$. Of course, one could also pick any $n\in
\NN$ and tensor this theory with another CFT of appropriate central
charge to give a total matter sector Virasoro element with $c_L=52$.

As an aside, it is interesting to note that the generalised
Nappi--Witten algebras $\mathfrak{nw}_{2n+2}$ are bargmannian
\cite{Figueroa-OFarrill:2022pus}, and sigma models constructed from
these are WZW models for strings propagating on bargmannian Lie
groups. Gauging the symmetry generated by the null element, $I$,
yields a new class of non-relativistic string models where the string
propagates on a Lie group with a bi-invariant galilean structure. The
full BRST quantisation of such string theories would then require the
gauging of the extension of the Virasoro algebra by this weight 1
primary field $I(z)$. This is precisely the algebra
$\hat{\g}_{\lambda=0}$, with $I(z)$ taking the role of $M(z)$. The
construction of such non-relativistic string models is part of ongoing
work.

By staring at \eqref{eq: Msug BMS3 from NW gen}, one might easily infer that it is actually possible to obtain realisations of $\hat{\g}_\lambda$ for all $\lambda\leq-1$ from $\widehat{\mathfrak{nw}}_{2n+2}$. Explicitly, this realisation is given by 
\begin{equation}
    T=T^{\text{sug}}, \quad M = I^{1-\lambda} \eqdef \underbrace{\normalord{I\dots I}}_{1-\lambda\text{ times}}.
\end{equation}
Naturally, one could also do this with a weight $(1,0)$ or $(0,1)$
$\beta\gamma$-systems and take normal-ordered products of the weight 1
field to construct $M$. Hence, in general, we can construct
$\g_{\lambda\leq -1}$ field theories out of $\g_{\lambda=0}$ field
theories. These are summarised in the following embedding diagram.
\begin{figure}[h!]
    \centering
        \begin{tikzcd}
                                &\shortstack{Weight $(1,0)$ \\ $(\beta, \gamma)$}&    \\
        \shortstack{$\hat{\mathfrak{g}}_{\lambda\leq-1}$ \\ $(T,M)$} \arrow[rr, "\big(T{,}\,\phi^{1-\lambda}\big)"] \arrow[ru, "\big(T^{\beta\gamma}{,}\,\beta^{1-\lambda}\big)"] \arrow[rd, "\big(T^\text{sug}{,}\,I^{1-\lambda}\big)",swap] &                            & \shortstack{$\hat{\mathfrak{g}}_{\lambda=0}$ \\ $(T,\phi)$}\arrow[lu, "\big(T^{\beta\gamma}{,}\,\beta\big)", swap]\arrow[ld, "\big(T^\text{sug}{,}\,I\big)"] \\
                                                          & \shortstack{$\widehat{\mathfrak{nw}}_{2n+2}$ \\ $(P^\pm_a,I,J)$} &        
    \end{tikzcd} 
    \caption{A diagram summarising the different explicit constructions of $\hat{g}_{\lambda\leq 0}$ field theories from weight $(1,0)$ $\beta\gamma$-systems and field theories with $\widehat{\mathfrak{nw}}_{2n+2}$ symmetry.}
    \label{fig: embedding summary}
\end{figure}
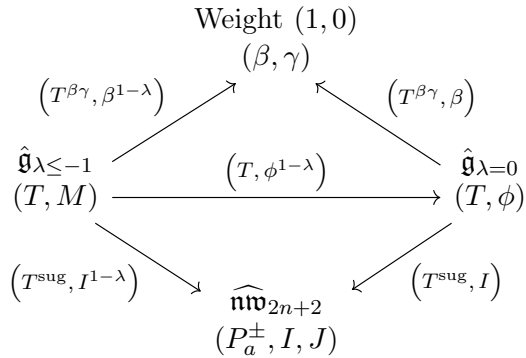

Coming back to the cases $\lambda=-1$ and $\lambda=0$, there exists a construction of a BMS\textsubscript{3} field theory with $c_M\neq 0$ out of a central extension of $\hat{\g}_{\lambda=0}$, given in \cite[Theorem 7.1]{Adamovic_2016}.

\section{Conclusions and future work}\label{sec: Conclusion}

We have shown that for any chiral $\hat{\g}_\lambda$ field theory,
\begin{enumerate}
\item There exists a free-field realisation in terms of a weight
  $(1-\lambda,\lambda)$ $\beta\gamma$-system with central charge $c_L
  = 2(6\lambda^2 - 6\lambda + 1)$.
\item There exists a free-field realisation in terms of a weight
  $(2,-1)$ $bc$-system and a weight $(1-\lambda,\lambda)$
  $BC$-system with central charge $c_L = - 26 - 2(6\lambda^2 -
  6\lambda + 1)$. This free-field realisation is the
  field-theoretic formulation of the semi-infinite wedge
  representation of $\hat{\g}_\lambda$ and is the ghost sector of
  the $\hat{\g}_\lambda$ field theory.
\item There exists a square-zero BRST operator if and only if the
  central of the matter sector $c_L = 26 + 2(6\lambda^2-6\lambda+1)$.
\end{enumerate}
Taking a closer at the case $\lambda=-1$, where $\hat{\g}_\lambda$ may
be further centrally extended to the BMS\textsubscript{3} algebra, we
have shown that for any extended \emph{chiral} CFT which admits
BMS\textsubscript{3} algebra symmetry, the BRST quantisation of such
field theories demands that the theory have central charge $(c_L,
c_M)=(52,0)$. This can be proved in two ways:
\begin{enumerate}
\item Algebraically formulate the BRST operator and Faddeev-Popov
  ghosts as the semi-infinite differential and the semi-infinite wedge
  representation of the BMS\textsubscript{3} algebra. The resulting
  ghost system forms a chiral BMS\textsubscript{3} field theory with
  $(c_L, c_M)=(-52,0)$. The vanishing of the total central charges,
  required for the BRST operator to be square-zero fixes the matter
  sector of the BMS\textsubscript{3} field theory to have
  $(c_L, c_M)=(52,0)$.
\item Take a completely field-theoretic approach and construct the
  most general chiral BMS\textsubscript{3} field theory from the
  $bc$-systems that appear as ghosts in the gauging procedure of
  BMS\textsubscript{3} symmetry. Doing so, one finds that obtaining a
  $c_M \neq 0$ realisation gives rise to quartic term that prevents
  the resulting BRST operator from being square-zero. It also forbids
  $M^\text{tot}$ from being BRST-exact, for all possible BRST
  operators. This once again forces the ghost sector to admit
  BMS\textsubscript{3} symmetry with $(c_L, c_M) = (-52,0)$, and we
  obtain the same conclusion as in the first approach.
\end{enumerate}

Now, can we still consider the notion of BRST quantisation of field
theories that admit BMS\textsubscript{3} algebra symmetry with
$c_M\neq 0$? As it stands, the answer is \emph{yes}. The findings here
only rule out the situations in which $c_M \neq 0$ is not possible,
namely chiral BMS\textsubscript{3} theories. Field theories which
admit BMS\textsubscript{3} symmetry in a manner that is not
holomorphically factorisable may still admit some consistent notion of
BRST quantisation with $c_M\neq 0$ through the formalism of
\subdefhighlight{full CFTs} \cite{MORIWAKI2023109125}. In particular,
we must consider that the BRST cohomology of such theories is not just
the semi-infinite cohomology of the underlying symmetry algebra of the
theory.

Another point to consider is the notion of ``flipped'' vacua in CFT
(e.g. \cite{Casali:2016atr, Hao:2021urq, Hwang:1998gs}). Such vacua
can be the starting points of physically valid constructions of
tensionless string spectra, as argued by the authors of
\cite{Bagchi:2020fpr}. More specifically, we need to pay attention to
the fact that the Virasoro automorphism $L_n \to -L_{-n},\ c\to -c$
does not lift to a VOA automorphism. Hence, without any further
assumptions, theories constructed as a result cannot be studied in a
rigorous manner using existing algebraic 2d CFT techniques.\footnote{In particular, CFTs with a normal vacuum in one sector and a flipped one in the other cannot be probed in this way. Intuitively, this is because boundedness conditions that naturally occur (e.g. highest weight or smooth modules of the underlying symmetry algebra) no longer exist when both of these vacua are put together in the same system.} We need an
alternative formalism (i.e., some sort of ``flipped'' VOA) to
rigorously encapsulate the modified normal-ordering with respect to
these flipped vacua. Perhaps such a formalism exists, using which one
can write down a different ``BRST quantisation'' procedure which
admits the existence of a square-zero BRST operator for $c_M \neq
0$. This would be particularly relevant to the case of tensionless
strings, because the BMS\textsubscript{3} symmetry not only emerges in
a non-chiral manner, but also in a way that mixes the positive and
negative modes of the two copies of the Virasoro algebra that appears
in the parent tensile closed string theory (i.e., an
``ultra-relativistic'' contraction \cite{Bagchi:2020fpr}).

Despite the aforementioned caveats preventing us from directly
applying our results to tensionless string theory, there do exist
string-theoretic realisations of the chiral BMS\textsubscript{3}
algebra. We have highlighted two such realisations in this paper:
\begin{enumerate}
    \item The ambitwistor string, given by \eqref{eq: ambitwistor Vir
        element} and \eqref{eq: ambitwistor M},
    \item The Nappi--Witten string, given by \eqref{eq: Tsug BMS3 from
        NW gen} and \eqref{eq: Msug BMS3 from NW gen}.
\end{enumerate}
A logical next step would be to seek other physical realisations of
this BMS\textsubscript{3} algebra, such as in terms of free bosons and
fermions. These would be intrinsic constructions, rather than as
limits of parent theories such as those considered in
\cite{Hao:2021urq, Yu:2022bcp, Hao:2022xhq}.

Naturally, one could also consider realising BMS\textsubscript{3}
using affine Kac-Moody currents. Sugawara constructions which are
compatible with Galilean contractions have been explored in
\cite{Rasmussen:2017eus, Ragoucy:2021iyo}, but again, we can look for
more general ones that need not necessarily be compatible with
contraction procedures. Doing so, one finds that although the end
product is a field-theoretic description of a Lie algebra
(i.e. BMS\textsubscript{3}), the conditions coming from this
construction are not Lie algebraic. In particular, $M$ need not be
built from an invariant tensor. Nonetheless, one could impose this as
a condition and then try classifying all the Lie algebras from which
one could build the BMS\textsubscript{3} algebra via the Sugawara
construction as a result. Some early progress in this regard, such as
the construction from the (generalised) Nappi--Witten algebra, looks
promising.

Finally, one could consider the BRST quantisation of
super-BMS\textsubscript{3} field theories. This could mean either the
minimally supersymmetric extension of a BMS\textsubscript{3} field
theory by a spin-$\nicefrac{3}{2}$ fermionic field or the algebra
obtained from the contraction of two copies of $N=1$ super-Virasoro
algebras \cite{Mandal:2016wrw}.

\section*{Acknowledgements}
JMF would like to acknowledge a useful conversation with Tim Adamo
about the ambitwistor string. JMF has spoken about this work in the
Erwin Schrödinger Institute, during his participation at the Thematic
Programme ``Carrollian Physics and Holography''. He would like to
thank the organisers, particularly Stefan Prohazka, for the invitation
and the hospitality. JMF has also spoken about this at the Niels Bohr
Institute and he would also like to thank Niels Obers and Emil Have for the
invitation and their hospitality. GSV would like to thank Arjun
Bagchi and David Ridout for insightful comments on an earlier draft of
this paper. GSV would also like to thank Ross McDonald, Christopher
Raymond and Ziqi Yan for thought-provoking discussions. GSV is
supported by a Science and Technologies Facilities Council studentship [grant
number 2615874].

\begin{appendices}

  \section{Proofs and calculations} \label{app: proofs and calculations}

  \subsection{Properties of meromorphic 2d CFTs}
  \label{app:prop-merom-2d}

Listed below are some key properties:
  \begin{enumerate}[label=(P\arabic*)]
    \item For all $A,B,C\in V$,
    \begin{equation*}
        \Bracket{A}{\Bracket{B}{C}{p}}{q} = \Koszulsign{A}{B} [B\Bracket{A}{C}{q}]_p + \sum_{l \geq 1} \binom{q-1}{l-1}\Bracket{\Bracket{A}{B}{l}}{C}{p+q-l} 
    \end{equation*}
    \label{item: Jac-identity 1}
    \item For all $A\in V$, $[A,-]_1$ is a super-derivation over all other $[-,-]_n$. That is,
    $$[A[BC]_n]_1 = [[AB]_1 C]_n  + \Koszulsign{A}{B}[B [AC]_1]_n.$$
    A special case of this is the derivation $[T,-]_1=\partial$
    \label{item: [A,-]_1 is a derivation over [-,-]_n}
    \item $[\partial A B]_n = -(n-1) [AB]_{n-1}$ and $[A \partial B]_n = (n-1)[AB]_{n-1} + \partial [AB]_{n}$. \label{item: proporties of partial in CFT}
    \item $(\partial A)_n = -(n+h_A) A_n$, where $A(z) = \sum_n A_n z^{-n-h_A}$ \label{item: modes of partial A}
    \item The brackets $[-,-]_n$ have conformal weight $-n$. That is,
    $$[-,-]_n\colon V_i \medotimes V_j \rightarrow V_{i+j-n}.$$ 
    \label{item: [-,-]_n weight -n}
    \item $[AB]_{n} = A_{n-h_A} B.$ \label{item: OPE to EndV acting on a state}
    \item The modes $\normalord{AB}_n$ in the expansion of the normal-ordered product of $A\in V_{h_A}$ and $B\in V_{h_B}$,   $\normalord{AB}(z) = \sum_n \normalord{AB}_n z^{-n-h_{A}-h_{B}}$, are given by
        \begin{equation*}
            \normalord{AB}_n = \sum_{l\leq -h_A} A_l B_{n-l} + \Koszulsign{A}{B} \sum_{l>-h_A} B_{n-l} A_l.
        \end{equation*}
        \label{item: Normal-ordered product mode exp}
    \item There exists a Lie superalgebra structure on the modes, given by
    \begin{equation*}
         [A_m, B_n] \defeq  A_m B_n - \Koszulsign{A}{B} B_n A_m = \sum_{l\geq 1} \binom{m+h_A-1}{l-1} \big([AB]_l\big)_{m+n}
    \end{equation*}
    \label{item: Mode algebra OPE relation}
  \end{enumerate}

\begin{proof}[Proof of \ref{item: Jac-identity 1} and \ref{item: [A,-]_1 is a derivation over [-,-]_n}]
First, we relabel $l\to l+q$ in the first summation of \ref{item: associativity CFT} and $l\to l+1$ in the second summation to rewrite the \ref{item: associativity CFT} as
    \begin{equation}
    \label{eq: relabelled associativity}
        [[AB]_p C]_q= \sum_{l\geq 0} (-1)^l \binom{p-1}{l}\left([A[BC]_{q+l}]_{p-l} + \Koszulsign[+p]{A}{B}[B[AC]_{l+1}]_{p+q-l-1} \right).
    \end{equation}
    Now consider the sum
    \begin{equation}
    \label{eq: associativity identity summed over}
        \sum_{l=0}^{p-1} \binom{p-1}{l}[[AB]_{p-l} C]_{q+l} = [[AB]_p C]_q + \binom{p-1}{1}[[AB]_{p-1} C]_{q+1} + \dots + [[AB]_1 C]_{p+q-1}.
    \end{equation}
    Using \eqref{eq: relabelled associativity}, we may write each term in the sum above as follows:
    \begin{equation*}
        \begin{split}
            [[AB]_p C]_q&= \sum_{k\geq 0} (-1)^k \binom{p-1}{k}\left([A[BC]_{q+k}]_{p-k} + \Koszulsign[+p]{A}{B}[B[AC]_{k+1}]_{p+q-k-1} \right)\\
            [[AB]_{p-1} C]_{q+1}&= \sum_{k\geq 0} (-1)^k \binom{p-2}{k}\left([A[BC]_{q+1+k}]_{p-1-k} - \Koszulsign[+p]{A}{B}[B[AC]_{k+1}]_{p+q-k-1} \right)\\
            &\vdots\\
            [[AB]_1 C]_{p+q-1}&=
            \left([A[BC]_{p+q-1}]_{1} - \Koszulsign{A}{B}[B[AC]_{1}]_{p+q-1} \right).
        \end{split}
    \end{equation*}
    Adding each term above after multiplying with the appropriate factor given in \eqref{eq: associativity identity summed over}, we notice that all terms, except for the $k=0$ and $k=p-1$ terms in the expansion of $[[AB]_p C]_q$, cancel out. Hence, we are left with
    \begin{equation}
        \sum_{l=0}^{p-1} \binom{p-1}{l}[[AB]_{p-l} C]_{q+l} = [A[BC]_q]_p - \Koszulsign{A}{B}[B[AC]_p]_q.
    \end{equation}
    Rearranging, we obtain \ref{item: Jac-identity 1}. Setting $q=1$ then proves \ref{item: [A,-]_1 is a derivation over [-,-]_n} right away.
\end{proof}

\begin{proof}[Proof of \ref{item: proporties of partial in CFT}]
From the definition of $\partial$,
\begin{align*}
    (\partial A)(z) B(w) \defeq \frac{d}{dz}A(z) B(w) 
    = \frac{d}{dz} \sum_{n\ll\infty} \frac{[AB]_n (w)}{(z-w)^n} 
    = \sum_{n\ll\infty} \frac{-n[AB]_n (w)}{(z-w)^{n+1}}.
\end{align*}
 But $(\partial A)(z) B(w)$ itself admits an OPE \[(\partial A)(z) B(w) = \sum_{n\ll\infty} \frac{[\partial A B]_n (w)}{(z-w)^n}.\] 
 Equating equal powers of $z-w$ gives \[ [\partial A B]_{n+1} = -n[AB]_n \iff [\partial A B]_n = -(n-1) [AB]_{n-1}.\]
 In a similar manner, or by using \ref{item: [A,-]_1 is a derivation over [-,-]_n} for $\partial = [T,-]_1$, we get $$[A \partial B]_n = (n-1)[AB]_{n-1} + \partial [AB]_{n}.$$
\end{proof}

\begin{proof}[Proof of \ref{item: modes of partial A}]
Using \ref{item: [A,-]_1 is a derivation over [-,-]_n} for $[T,-]_1$, we have
$$[T\partial A]_2=[TA]_1+\partial [TA]_2 = (h_A + 1)\partial A.$$
Hence, $(\partial A)(z)$ admits a mode expansion
\begin{align*}
    (\partial A)(z) = \sum_n (\partial A)_n z^{-n-(h_A + 1)}.
\end{align*}
On the other hand, \[(\partial A)(z) \defeq \frac{d}{dz} A(z)
    = \frac{d}{dz} \sum_n A_n z^{-n-h_A}
    = \sum_n -(n+h_A) A_n z^{-n-h_A-1}.\]
Equating the two mode expansions gives $(\partial A)_n = -(n+h_A)A_n.$
\end{proof}

\begin{proof}[Proof of \ref{item: [-,-]_n weight -n}]
Let $A\in V_{h_{A}}$ and $B\in V_{h_{B}}$. By \ref{item: D5 Vir element}, this means
$[TA]_2 = h_A A$ and $[TB]_2 = h_B B$. Thus,
\begin{align*}
    [T[AB]_n]_2 &= [A[TB]_2]_n + [[TA]_1 B]_{n+1} + [[TA]_2 B]_n  \tag*{by \ref{item: Jac-identity 1}}\\
    &= h_B [AB]_n + [\partial A B]_{n+1} + h_A [AB]_n \tag*{by \ref{item: D5 Vir element}}\\
    &= (h_A + h_B - n)[AB]_n \tag*{by \ref{item: proporties of partial in CFT}}.
\end{align*}
This proves that $[-,-]_n$ indeed has conformal weight $-n$. 
\end{proof}

\begin{proof}[Proof of \ref{item: OPE to EndV acting on a state}]
We have, by \ref{item: Identity axiom CFT},
\[\lim_{w\to0} A(z) B(w) \mathbbm{1} = A(z) B = \sum_n z^{-n-h_A} A_n B.\] 
At the same time,
\[\lim_{w\to0} \sum_n \frac{[AB]_n(w)}{(z-w)^n} \mathbbm{1} = \sum_n z^{-n} [AB]_n.\]
Equating equal powers of $z$ in the two expressions above gives $[AB]_n = A_{n-h_A} B$ as desired.
\end{proof}

\begin{proof}[Proof of \ref{item: Normal-ordered product mode exp}]
Using \ref{item: associativity CFT}, we can write $[\normalord{AB}C]$ as
\begin{equation}
\label{eq: bracket endo step 1}
    [\normalord{AB} C]_q = \sum_{l\geq q} [A[BC]_l]_{q-l} + \Koszulsign{A}{B} \sum_{l\geq1} [B[AC]_l]_{q-l}.
\end{equation}
Using properties \ref{item: [-,-]_n weight -n} and \ref{item: OPE to EndV acting on a state},
\begin{equation*}
\begin{split}
    &[\normalord{AB} C] = \normalord{AB}_{q-h_A-h_B} C.\\
    &[A[BC]_l]_{q-l} = A_{q-l-h_A}B_{l-h_B}C.\\
    &[B[AC]_l]_{q-l} = B_{q-l-h_b} A_{l-h_A}C.
\end{split}
\end{equation*}
Substituting these back into \eqref{eq: bracket endo step 1} gives
\begin{equation*}
    \normalord{AB}_{q-h_A-h_B} C = \left(\sum_{l\geq q} A_{q-l-h_A}B_{l-h_B} + \Koszulsign{A}{B} \sum_{l\geq 1} B_{q-l-h_b} A_{l-h_A} \right)C, \quad \forall C \in \M.
\end{equation*}
We may abstract $C$ since it holds true $\forall C \in \M$. Relabelling the first summation with $m=q-l-h_A$ and letting $n\defeq q-h_A-h_B$ gives
\[\sum_{l\geq q} A_{q-l-h_A}B_{l-h_B} = \sum_{m\leq -h_A} A_{m} B_{n-m} = \sum_{l\leq -h_A} A_{l} B_{n-l} .\]
Relabelling the second summation with $m = l-h_A$ gives
\[ \sum_{l\geq 1} B_{q-l-h_b} A_{l-h_A} = \sum_{m\geq -h_A+1} B_{n-m}A_m = \sum_{l> -h_A} B_{n-l} A_l.\]
Putting them back together gives us the desired result.
\end{proof}

\begin{proof}[Proof of \ref{item: Mode algebra OPE relation}]
For all $C\in V$, we may act $[A_m, B_n]\in \End V$ to get
\begin{align*}
    [A_m, B_n] C &\defeq  A_m (B_n C )- \Koszulsign{A}{B} B_n (A_m C)\\
    &= A_m \Bracket{B}{C}{n+h_B} - \Koszulsign{A}{B} B_n \Bracket{A}{C}{m+h_A} \tag*{by \ref{item: OPE to EndV acting on a state}}\\
    &= \Bracket{A}{\Bracket{B}{C}{n+h_B}}{m+h_A} - \Koszulsign{A}{B} \Bracket{B}{\Bracket{A}{C}{m+h_A}}{n+h_B}\tag*{by \ref{item: OPE to EndV acting on a state}}\\
    &= \sum_{l\geq  1} \binom{m+h_A - 1}{l-1} \Bracket{\Bracket{A}{B}{l}}{C}{m+n+h_A + h_B -l} \tag*{by \ref{item: Jac-identity 1}} \\
    &= \sum_{l\geq  1} \binom{m+h_A - 1}{l-1} \big(\Bracket{A}{B}{l}\big)_{m+n}C \tag*{by \ref{item: [-,-]_n weight -n} and \ref{item: OPE to EndV acting on a state}}. 
\end{align*}
Since it holds for any $C \in V$, we obtain the desired result. 
\end{proof}

\subsection{Some proofs in semi-infinite cohomology}
\label{app:some-proofs-semi}

\begin{proof}[Proof of Proposition~\ref{prop:cliff-module}]
Define $\kappa: \text{Cl}(\g\medoplus\g') \rightarrow \End(\semiinfforms)$ via $\kappa(x+x') \defeq \iota(x) +\varepsilon(x')$. We need to show that $(\kappa(x+x'))^2 = (x+x')\cdot (x+x') = \langle x',x \rangle \Id_{\semiinfforms}$. 
\begin{equation*}
         (\kappa(x+x'))^2 = \iota(x)^2 +\varepsilon(x')^2 + \iota(x)\varepsilon(x') + \varepsilon(x')\iota(x) =[\iota(x),\varepsilon(x')] = \langle x', x \rangle \Id_{\semiinfforms},
\end{equation*}
where the last equality follows from lemma \ref{lem: fundamental_anticomms}.
\end{proof}

\begin{proof}[Proof of Proposition \ref{prop: rho_intext_commrels}]
We will perform calculations on monomials and the argument extends to all semi-infinite forms by $\mathbb{C}$-linearity.
\begin{equation*}
    \begin{aligned}
        [\rho(x), \varepsilon(y')] e'_{i_1} \wedge e'_{i_2} \wedge \dots 
          &=  \ad'_x y' \wedge e'_{i_1} \wedge e'_{i_2} \wedge \dots  + \sum_{k\geq 1} y'\wedge e'_{i_1}\wedge e'_{i_2} \wedge \dots \wedge \ad'_x e'_{i_k} \wedge \dots\\
          &\ - y' \wedge \sum_{k\geq 1} e'_{i_1} \wedge e'_{i_2} \wedge \dots \wedge \ad'_x e'_{i_k} \wedge \dots \\
          &= \ad'_x y'\wedge e'_{i_1} \wedge e'_{i_2} \wedge \dots\\
          &= \varepsilon(\ad'_x y') e'_{i_1} \wedge e'_{i_2} \wedge \dots \\[1em]
         [\rho(x), \iota(y)] e'_{i_1} \wedge e'_{i_2} \wedge \dots 
            &= \rho(x) \sum_{k\geq 1} (-1)^{k-1} \langle y,e'_{i_k} \rangle e'_{i_1} \wedge e'_{i_2} \wedge \dots \wedge \widehat{e'_{i_k}} \wedge \dots\\
            &\ - \iota(y) \sum_{k\geq1} e'_{i_1} \wedge e'_{i_2} \wedge \dots \wedge \ad'_x e'_{i_k} \wedge \dots\\
            &= \sum_{k\geq 1} (-1)^{k-1} -\langle y, \ad'_x e'_{i_k} \rangle e'_{i_1} \wedge e'_{i_2} \wedge \dots \wedge \widehat{\ad'_x e'_{i_k}} \wedge \dots \\
            &= \sum_{k\geq 1} (-1)^{k-1} \langle \ad_x y, e'_{i_k} \rangle e'_{i_1} \wedge e'_{i_2} \wedge \dots \wedge \widehat{e'_{i_k}} \wedge \dots\\
            &= \iota(\ad_x y) e'_{i_1} \wedge e'_{i_2} \wedge \dots
    \end{aligned}
\end{equation*}
\end{proof}

\begin{proof}[Proof of Proposition \ref{prop: FGZ_prop_1.1}]
Naively, the failure of $\rho\colon\g\rightarrow\semiinfforms$ to be a representation, given by $[\rho(x), \rho(y)] - \rho([x,y]) = \gamma(x,y)$, is encapsulated by some bilinear form $\gamma\colon\g\times\g \to \CC$ which is non-zero only if $x\in\g_n$ and $y\in\g_n$, for all $n\in\ZZ$. We deduce the fact that $\gamma(x,y)=\gamma(y,x)$ from the antisymmetry of the commutator and Lie brackets on the LHS. The fact that  is encapsulated by a 2-cocycle in $\gamma$ follows from the Jacobi identity of both $\g$ and the associative algebra on $\End\semiinfforms$ given by the commutator bracket. 
    \begin{equation}
    \begin{split}
        0&= [\rho(x), [\rho(y),\rho(z)]] - \rho([x,[y,z]]) +  [\rho(y), [\rho(z),\rho(x)]] - \rho([y,[z,x]])\\
         &+  [\rho(z), [\rho(x),\rho(y)]] - \rho([z,[x,y]]) \\
         &= \gamma(x,[y,z]) + \gamma(y,[z,x]) + \gamma(z,[x,y])\\
    \end{split}
    \end{equation}
    Hence, $\gamma$ obeys the cocycle condition.
    The fact that this failure is only important up to an equivalence class in cohomology is reinforced in the second statement of proposition \ref{prop: FGZ_prop_1.1}, proved as follows:
    If $\gamma$ is a coboundary, there exists $\alpha\in\g'$ such that $\gamma = \partial\alpha$ and recall that $\gamma(x,y) = \partial\alpha (x,y) = -\alpha([x,y])$. In particular, this means that $\alpha\in\g'_0$. We currently have a $\rho\colon \g \rightarrow \End(\semiinfforms)$ that obeys \ref{prop: FGZ_prop_1.1}. Let us define a new representation $\Tilde{\rho}\colon \g \rightarrow \End(\semiinfforms)$ given by $\Tilde{\rho}(x) \defeq \rho(x) - \langle \alpha, x \rangle$. Then
\begin{align*}
    [\Tilde{\rho}(x), \Tilde{\rho}(y)] 
        &\defeq [\rho(x) - \langle \alpha, x \rangle, \rho(y) - \langle \alpha, y \rangle]\\
        &\ =[\rho(x), \rho(y)]\\
        &\ =\rho([x,y]) + \gamma(x,y) \tag*{(by proposition \ref{prop: FGZ_prop_1.1})} \\
        &\ =\rho([x,y]) + \partial\alpha(x,y) \tag*{(by definition of $\gamma$)} \\
        &\ =\rho([x,y]) - \langle \alpha, [x,y] \rangle \\  
        &\ =\Tilde{\rho}([x,y]) \tag*{(by definition of $\Tilde{\rho}$)}.
\end{align*}
Any $\omega_0$ defines $\rho\colon \g \rightarrow \semiinfforms$ satisfying proposition \ref{prop: FGZ_prop_1.1} using some choice of $\beta$. If $\exists \alpha \in \g_0$ such that $\gamma = \partial \alpha$, then one can make the modification $\beta \rightarrow \Tilde{\beta} \defeq \beta - \alpha$ so that $\gamma = 0$.
\end{proof}

\begin{proof}[Proof of Theorem \ref{thm: semiinfrep of g_lambda central charge}]
We first recall our choice of basis \eqref{eq: g_lambda basis e_i}, vacuum \eqref{eq: g_lambda vacuum} and the statement of lemma \ref{lem: rho(L_n) and rho(M_n) semiinfforms}. Choosing $n>0$, it follows that
\begin{equation*}
        [\rho(L_n), \rho(L_{-n})]\omega_0 = \big(\rho(L_n) \rho(L_{-n}) - \rho(L_{-n}) \rho(L_n)\big)\omega_0 = \rho(L_n) \rho(L_{-n}) \omega_0.
\end{equation*}
First, we simplify $\rho(L_{-n})\omega_0$ using \eqref{eq: normal-ordered product SIC} as follows:
\begin{equation*}
\begin{split}
    \rho(L_{-n})\omega_0 = \bigg(-&\sum_{i\leq 1} (n+i) \iota(L_{i-n}) \varepsilon(L'_{i}) - \sum_{i\leq 1} (i-\lambda n) \iota(M_{i-n})\varepsilon (M'_i)  \\
    +&\sum_{i>1} (n+i) \varepsilon(L'_{i}) \iota(L_{i-n}) + \sum_{i>1} (i-\lambda n) \varepsilon(M'_i)\iota(M_{i-n})\bigg)\omega_0 \\
    =&\bigg(\sum_{i=2}^{n+1} (n+i) \varepsilon(L'_{i})\iota(L_{i-n}) + \sum_{i=2}^{n+1} (i-\lambda n) \varepsilon(M'_i)\iota(M_{i-n})\bigg)\omega_0
\end{split}
\end{equation*}
Thus, as expected, the infinite sums in $\rho(L_{-n})$ truncate to finite ones when acting on the vacuum $\omega_0$.
Splitting each normal-ordered term in $\rho(L_n)$ according to the normal-ordering prescription \eqref{eq: normal-ordered product SIC}, we may expand $\rho(L_n)\rho(L_{-n})\omega_0$ into the 8 terms as done below:
\begin{equation*}
    \begin{split}
        \rho(L_n)\rho(L_{-n})\omega_0 &= \sum_{j\leq 1} \sum_{i=2}^{n+1} (n-j)(n+i) \iota(L_{j+n})\varepsilon(L'_j) \varepsilon(L'_{i})\iota(L_{i-n}) \omega_0\\
        &+\sum_{j\leq 1} \sum_{i=2}^{n+1} (n-j)(i-\lambda n)  \iota(L_{j+n})\varepsilon(L'_j) \varepsilon(M'_i)\iota(M_{i-n})\omega_0\\
        &-\sum_{j\leq 1} \sum_{i=2}^{n+1} (j+\lambda n) (n+i)\iota(M_{j+n})\varepsilon (M'_j) \varepsilon(L'_{i})\iota(L_{i-n})\omega_0\\
        &-\sum_{j\leq 1} \sum_{i=2}^{n+1} (j+\lambda n)(i-\lambda n) \iota(M_{j+n})\varepsilon (M'_j) \varepsilon(M'_i)\iota(M_{i-n})\omega_0\\
        &-\sum_{j>1}\sum_{i=2}^{n+1} (n-j)(n+i) \varepsilon(L'_j)\iota(L_{j+n})  \varepsilon(L'_{i})\iota(L_{i-n}) \omega_0\\
        &-\sum_{j>1}\sum_{i=2}^{n+1} (n-j)(i-\lambda n) \varepsilon(L'_j)\iota(L_{j+n}) \varepsilon(M'_i)\iota(M_{i-n})\omega_0\\
        &+\sum_{j>1}\sum_{i=2}^{n+1} (j+\lambda n)(n+i) \varepsilon(M'_j)\iota(M_{j+n}) \varepsilon(L'_{i})\iota(L_{i-n})\omega_0\\
        &+\sum_{j>1}\sum_{i=2}^{n+1} (j+\lambda n)(i-\lambda n) \varepsilon(M'_j)\iota(M_{j+n}) \varepsilon(M'_i)\iota(M_{i-n})\omega_0.
    \end{split}
\end{equation*}
Lines 2, 3, 6 and 7 vanish because the annihilation operators coming from the expansion of $\rho(L_n)$ can freely (up to a sign) commute past the other operators to act on the vacuum without the addition of any other non-trivial terms. Lines 5 and 8 are also vanishing because the non-trivial contribution we get from commuting the annihilation operators past the others is a $\delta_{j,-n+i}$ term, which is never non-zero for the values $i$ and $j$ take in those sums. Thus, the only non-zero contributions are from lines 1 and 4. After perfoming the necessary commutations, we are left with
\begin{equation*}
\begin{split}
    \rho(L_n)\rho(L_{-n})\omega_0 &= \sum_{j\leq 1} \sum_{i=2}^{n+1} -(n-j)(n+i) \iota(L_{j+n})\varepsilon(L'_{i})\delta_{j,i-n} \omega_0\\
    &+\sum_{j\leq 1} \sum_{i=2}^{n+1} (j+\lambda n)(i-\lambda n)\iota(M_{j+n}) \varepsilon(M'_i)\delta_{j,i-n}\omega_0\\
    &=\sum_{i=2}^{n+1} \Big((i-2n)(n+i)\iota(L_i)\varepsilon(L'_i)+\big(i+(\lambda-1) n\big)(i-\lambda n)\iota(M_{i}) \varepsilon(M'_i)\Big)\omega_0.\\
    &=\sum_{i=2}^{n+1} \Big((i-2n)(n+i)+\big(i+(\lambda-1) n\big)(i-\lambda n)\Big)\omega_0.
\end{split}
\end{equation*}
On the other hand, $\rho([L_n, L_{-n}])\omega_0=2n\rho(L_0)\omega_0=0$. Thus, any non-zero contribution to $[\rho(L_n), \rho(L_{-n})]\omega_0$ is either from a coboundary term (i.e. a different choice of $\beta$) which, according to the form of \eqref{eq: rho_semiinfforms}, must be proportional to $L'_0$ or a cohomologically non-trivial cocycle term which implies that we would need to centrally extend our Lie algebra to make $\rho$ a representation.
This non-zero contribution is precisely the finite sum we have obtained above, which we now evaluate:
\begin{equation*}
\begin{split}
    \sum_{i=2}^{n+1} \Big((i-2n)(n+i)+\big(i+(\lambda-1) n\big)(i-\lambda n)\Big)
    =&\sum_{i=2}^{n+1}\Big(2i^2-2ni-\big(2+\lambda(\lambda-1)\big)n^2\Big)\\
    =&\frac{2}{3}n^3+3n^2+\frac{13}{3}n-n^3-3n^2-\big(2+\lambda(\lambda-1)\big)n^3\\
    =&\left(-\frac{1}{3}-2-\lambda^2+\lambda\right)n^3+\frac{13}{3}n\\
    =&\left(-\frac{7}{3}-\lambda^2+\lambda\right)n^3+\frac{13}{3}n\\
    =&\left(-\frac{7}{3}+\frac{1}{6}-\frac{1}{6}-\lambda^2+\lambda\right)n^3+\frac{13}{6}n+\frac{13}{6}n\\
    =&-\frac{13}{6}(n^3-n)+\frac{1}{6}(-1-6\lambda^2+6\lambda)n^3+\frac{13}{6}n\\
    =&-\frac{13}{6}(n^3-n)-\frac{1}{12}\big(2(6\lambda^2-6\lambda+1)\big)n^3+\frac{13}{6}n\\
    &+\frac{1}{12}\big(2(6\lambda^2-6\lambda+1)\big)n-\frac{1}{12}\big(2(6\lambda^2-6\lambda+1)\big)n\\
    =&-\frac{1}{12}\big(26+2(6\lambda^2-6\lambda+1)\big)(n^3-n)\\
    &+ \frac{1}{12}(26-2(6\lambda^2-6\lambda+1)n.
\end{split}
\end{equation*}
Thus, we have manipulated $[\rho(L_n),\rho(L_{-n})]\omega_0$ into the above form containing two terms. The first terms is proportional to the Gelfand-Fuks cocycle given by \ref{eq: Gelfand-Fuks cocycle}, while the second term is proportional to $n$. The presence of the first term indicates that we need to centrally extend $\g_\lambda$ to $\hat{\g}_\lambda$ using the Gelfand-Fuks cocycle in order to make $\rho$ a Lie algebra representation on $\semiinfforms$. The proportionality factor is precisely the action of the new central charge on $\semiinfforms$, i.e. 
\begin{equation}
\label{eq: rho(c_L) computed}
    \rho(c_L) = -\big(26+2(6\lambda^2-6\lambda+1)\big)\Id_{\semiinfforms}.
\end{equation}
The contribution proportional to $n$ can be absorbed by modifying our initial naive choice of $\beta=0$ to
\begin{equation}
\label{eq: new_beta}
    \beta = \frac{1}{12}(13-6\lambda^2+6\lambda-1)L'_0.
\end{equation}
For brevity, we reiterate the fact that $\rho\colon\hat{\g}_\lambda\rightarrow\semiinfforms$ defines a representation of $\hat{\g}_\lambda$ on $\semiinfforms$:
\begin{equation*}
\begin{split}
    &[\rho(L_n),\rho(L_{-n})]\omega_0 = \rho([L_n, L_{-n}])\omega_0\\
    \iff &-\frac{1}{12}\big(26+2(6\lambda^2-6\lambda+1)\big)(n^3-n)+ \frac{1}{12}(26-2(6\lambda^2-6\lambda+1)n = 2n\rho(L_0)\omega_0 +\frac{1}{12}n(n^2-1)\rho(c_L).
\end{split}
\end{equation*}
Substituting \eqref{eq: rho(c_L) computed} and \eqref{eq: new_beta} above, we see an agreement of the LHS and RHS, demonstrating the validity of our calculations. For generic $\lambda\in\ZZ$, in particular, for $\lambda \neq -1,0,1$, $\g_\lambda$ has no other cohomologically non-trivial 2-cocycles. Hence, it suffices to check the failure of $\rho$ on just the $L_n$ generators.

On the other hand, when $\lambda=-1,0,1$, we need to perform additional checks on other pairs of generators. Taking $\lambda=-1$, we may repeat the exact calculation above for the centreless BMS\textsubscript{3} algebra. As expected, we get $\rho(c_L) = -52$, which is in agreement with the general case \eqref{eq: rho(c_L) computed}. However, $\g_{\lambda=-1}$ admits a second cohomologically non-trivial 2-cocycle 
\begin{equation*}
    \gamma_M(L_n, M_m) = \frac{1}{12}n(n^2-1)\delta^0_{m+n}.
\end{equation*}
Hence, we need to check if this 2-cocycle $\gamma_M$ is required to ensure that $\rho$ does not fail as a representation on $\semiinfforms$. Once again, choosing the same basis \eqref{eq: g_lambda basis e_i}, vacuum \eqref{eq: g_lambda vacuum} and $\beta$ as given in \eqref{eq: new_beta} with $\lambda=-1$, 
\begin{equation*}
    [\rho(L_n), \rho(M_{-n})]\omega_0 = \big(\rho(L_n) \rho(M_{-n}) - \rho(M_{-n}) \rho(L_n)\big)\omega_0 = \rho(L_n) \rho(M_{-n}) \omega_0
\end{equation*}
Using lemma \ref{lem: rho(L_n) and rho(M_n) semiinfforms} and \eqref{eq: normal-ordered product SIC},
\begin{equation*}
\begin{split}
    \rho(M_{-n})\omega_0 
        &= \left(-\sum_{i\leq 1} (n+i) \iota(M_{i-n}) \varepsilon(L'_i) + \sum_{i>1} (n+i)\varepsilon(L'_i)\iota(M_{i-n}) \right)\omega_0\\
        &= \sum^{n+1}_{i=2}(n+i)\varepsilon(L'_i)\iota(M_{i-n}) \omega_0.
\end{split}
\end{equation*}
For clarity, we explicitly write out the 4 terms in $\rho(L_n)\rho(M_{-n})\omega_0$:
\begin{equation*}
    \begin{split}
        \rho(L_n)\rho(M_{-n})\omega_0 &= \sum_{j\leq1}\sum^{n+1}_{i=2}(n-j)(n+i) \iota(L_{j+n})\varepsilon(L'_j) \varepsilon(L'_i)\iota(M_{i-n})\omega_0\\
        &-\sum_{j>1}\sum^{n+1}_{i=2} (n-j)(n+i) \varepsilon(L'_j)\iota(L_{j+n})\varepsilon(L'_i)\iota(M_{i-n})\omega_0\\
        &-\sum_{j\leq1}\sum^{n+1}_{i=2}(j-n)(n+i)\iota(M_{j+n})\varepsilon (M'_j)\varepsilon(L'_i)\iota(M_{i-n})\omega_0\\
        &+\sum_{j>1} \sum^{n+1}_{i=2}(j-n)(n+i) \varepsilon (M'_j)\iota(M_{j+n})\varepsilon(L'_i)\iota(M_{i-n})\omega_0
    \end{split}
\end{equation*}
By the same arguments, we are only left with line 3:
\begin{equation*}
    \rho(L_n)\rho(M_{-n})\omega_0 =\sum^{n+1}_{i=2}\big(i-2n\big)(n+i)\iota(M_i)\varepsilon(L'_i)\omega_0.
\end{equation*}
However, this time, the RHS is zero since the contraction and wedge operations that appear in the RHS are not canonically dual to each other, and thereby anti-commute freely. Thus, we do not need to modify $\g_{\lambda=-1}$ through the addition of the second non-trivial 2-cocycle $\gamma_M$ to make $\rho$ a representation on $\semiinfforms$. One can perform identical calculations with $\lambda=0$ and $\lambda=1$ as well, since those are the only other values for $\lambda$ which $\dim H^2(\g_\lambda)>1$. It then follows that $\rho\colon\hat{\g}_\lambda\rightarrow\End\semiinfforms$ indeed defines a Lie algebra representation for all $\lambda\in\ZZ$.
\end{proof}

\begin{proof}[Proof of Theorem \ref{thm: BRST quantisation of ghat_lambda field theories}]
  As mentioned earlier, the computation of the square of the BRST
  operator is one of the most prominent examples of the computational
  power of the field-theoretic formulation of semi-infinite
  cohomology. To compute $d^2$ would require the simplification of the
  product of two infinite sums, each of which is a nested infinite sum
  of products of the modes $b_n,\ c_n,\ B_n$ and $C_n$. Such an
  immensely tedious calculation is greatly simplified as follows. We
  first notice that for any $Y \in \M \medotimes \semiinfforms$,
  \begin{align*}
      d^2 Y &= [\mathfrak{J}[\mathfrak{J} Y]_1]_1 \tag*{by \ref{item: OPE to EndV acting on a state}} \\
        &=-[\mathfrak{J}[\mathfrak{J} Y]_1]_1 + \sum_{l\geq 1}\binom{0}{l-1}[[\mathfrak{J} \mathfrak{J}]_l Y]_{2-l} \tag*{by \ref{item: Jac-identity 1}}
  \end{align*}
  and hence $d^2 Y  = \tfrac12 [[\mathfrak{J}\, \mathfrak{J}]_1\,
  Y]_1$. 
  Computing the OPE of $\mathfrak{J}$ with itself, we get
    \begin{equation}
        \label{eq: jBRSTjBRST first order pole}        [\mathfrak{J}\mathfrak{J}]_1=\frac{1}{2}(1+\lambda)\partial\normalord{McC}+\left(\frac{7}{4}-\frac{3\lambda}{4}+\frac{\lambda^2}{2}\right)\normalord{\partial^2 c \partial c} + \left(-\frac{7}{12}+\frac{c^{mat}}{12}+\frac{\lambda}{4}-\frac{\lambda^2}{2}\right)\normalord{\partial^3 c c}.
    \end{equation}
    We demand that $\ker\big([[\mathfrak{J}\mathfrak{J}]_1,-]_1\big)=\M \medotimes \semiinfforms$ by enforcing that this map is zero on all generators of the $\hat{\g}_\lambda$ field theory. Doing so enforces $c^{mat}=4(7-3\lambda+3\lambda^2)=26+2(6\lambda^2-6\lambda+1)$. This is exactly in agreement with the statement of lemma \ref{lem: rel. SIC requirement}. Notice that this makes $[\mathfrak{J}\mathfrak{J}]_1$ a total derivative, so that by \ref{item: proporties of partial in CFT}, $d^2 Y = [[\mathfrak{J}\mathfrak{J}]_1 Y]_1 = 0$ for all $Y\in \M\medotimes\semiinfforms$ indeed.
\end{proof}
  
\end{appendices}

\printbibliography

\end{document}